\documentclass[11pt]{article}
\usepackage{amssymb}
\usepackage{amsmath}
\usepackage{amsfonts}
\usepackage{geometry}
\usepackage{amsfonts}
\usepackage[latin9]{inputenc}
\usepackage{geometry}
\usepackage{setspace}
\usepackage{dcolumn}
\usepackage{graphicx}
\usepackage[position=bottom]{subfig}
\usepackage[authoryear]{natbib}
\usepackage[colorlinks,linkcolor=blue,citecolor=blue]{hyperref}

\newcolumntype{d}{D{.}{.}{-1}}
\numberwithin{equation}{section}
\newtheorem{theorem}{Theorem}

\newtheorem{corollary}[theorem]{Corollary}

\newtheorem{definition}[theorem]{Definition}

\newtheorem{lemma}{Lemma}

\numberwithin{equation}{section}
\numberwithin{lemma}{section}
\numberwithin{theorem}{section}
\newenvironment{proof}[1][Proof]{\noindent \textbf{#1.} }{\  \rule{0.5em}{0.5em}}

\begin{document}

\title{Closed-form approximations of moments and densities of
continuous--time Markov models\thanks{{\scriptsize We would like to thank
participants at the NBER--NSF\ Time Series Conference 2022 and Seungmoon
Park for valuable comments and suggestions}\texttt{\scriptsize .}}}
\author{Dennis Kristensen\thanks{{\scriptsize {}Department of Economics,
University College London; E-mail: }\texttt{{\scriptsize %
{}d.kristensen@ucl.ac.uk.}}} \and Young Jun Lee\thanks{{\scriptsize {}IGIER,
Universit� Bocconi; E-mail: }\texttt{{\scriptsize {}young.lee@unibocconi.it.}%
}} \and Antonio Mele\thanks{{\scriptsize USI Lugano, {}Swiss Finance
Institute and CEPR; E-mail: }\texttt{{\scriptsize {}antonio.mele@usi.ch}}} }
\maketitle

\begin{abstract}
This paper develops power series expansions of a general class of moment
functions, including transition densities and option prices, of
continuous-time Markov processes, including jump--diffusions. The proposed
expansions extend the ones in \cite{kristensen2011} to cover general Markov
processes. We demonstrate that the class of expansions nests the transition
density and option price expansions developed in \cite{Yangetal2019} and 
\cite{Wan2021} as special cases, thereby connecting seemingly different
ideas in a unified framework. We show how the general expansion can be
implemented for fully general jump--diffusion models. We provide a new
theory for the validity of the expansions which shows that series expansions
are not guaranteed to converge as more terms are added in general. Thus,
these methods should be used with caution. At the same time, the numerical
studies in this paper demonstrate good performance of the proposed
implementation in practice when a small number of terms are
included.\bigskip\ 

\medskip {}

\end{abstract}

\newpage

\section{Introduction\label{sec: intro}}

Continuous-time jump-diffusion processes are used in economics and finance
to model the dynamics of state variables (see, e.g., \citealp{bjork2008}).
They lead to a simple and elegant analysis of problems such as the pricing
of financial assets, portfolio management and other dynamic phenomena. This
comes at a big computational cost though: Many relevant characteristics,
such as moments and densities, of such processes cannot be expressed in
closed-form except in a few special cases. This hampers their practical use
and implementation. This has led researchers to develop numerical methods
for the computation of these. Broadly speaking, these methods fall in three
categories: Finite--difference methods \citep[][]{Ames1992},
simulation--based methods 
\citep[see, e.g.,][]{elerian2001,
brandt2002, durham2002,beskos2009,kristensen2012,sermaidis2013} and series
expansions 
\citep[see,
e.g.,][]{aitsahalia2002, bakshi2006, yu2007, aitsahalia2008, Filipovic2013, li2013}%
. This paper focuses on the latter category.

Most existing expansions proposed in the literature are application
specific: Depending on the particular features of the chosen moment and
model of the underlying stochastic process, different methods have been
developed. One exception is \cite{kristensen2011} who developed power series
expansions that covered a general class of moment functions and the
transition density of multivariate diffusion processes. Their focus was on
applications to option pricing but the class of expansions applies more
generally. The current paper makes four contributions:

First, we demonstrate that the class of series expansions of \cite%
{kristensen2011} are easily extended to cover fully general continuous--time
Markov models, including any jump--diffusion process. Thus, the proposal of 
\cite{kristensen2011} can in principle be applied to any moment of any
Markov process. As part of this extension, we present a novel derivation and
representation of the series expansion of \cite{kristensen2011}. This new
representation highlights important features of the original expansion that
was perhaps not obvious from the analysis of \cite{kristensen2011}.

Second, we revisit the recent work of \cite{Yangetal2019} and \cite{Wan2021}
and demonstrate that in fact their proposed expansions of transition
densities and option prices are special cases of \cite{kristensen2011}.
Thus, at a theoretical level the expansions in \cite{Yangetal2019} and \cite%
{Wan2021} are not new. At the same time, it should be emphasized that \cite%
{Yangetal2019} and \cite{Wan2021} make important contributions in terms of
the practical implementation of the proposal in \cite{kristensen2011}. They
develop numerical algorithms that allow for fast implementation of the
general method of \cite{kristensen2011} when applied to transition densities
and option prices of diffusion processes and a limited set of
jump--diffusion processes. As such, the current paper should hopefully
clarify the relationship between these three existing papers and their
relative contributions to the literature.

Third, we propose a novel numerical implementation of our series expansions
when applied to general jump--diffusion models. The algorithms of \cite%
{kristensen2011} and \cite{Yangetal2019} are restricted to pure diffusions
while the extension found in \cite{Wan2021} requires the jump component to
be fully independent of the diffusive component. That is, the jump intensity
and the jump sizes are not allowed to be state--dependent. Our numerical
implementation allows for both to be state--dependent. We demonstrate
through a series of numerical studies that our numerical method works well
in practice.

Fourth, we provide a novel theory for the validity of power series
expansions of moment functions of continuous--time Markov processes used
here and elsewhere in the literature, including all above references to
papers employing series--based approximations. Most existing theoretical
results for these expansions only show that a given moment expansion
converges as the time interval over which the conditional moment is defined
shrinks to zero. As such existing results provide no guarantees that the
approximation error will get smaller as more terms are added to expansion;
in fact, nothing rules out that the approximation error may actually explode
as more terms are added. For the power series expansion to be reliable, it
is desirable with conditions under which the expansions converge not only
over shrinking time intervals but also over a fixed time interval. We here
provide guarantees for the approximations to be numerically stable as the
order of the approximation grows. Our theoretical results rely on
semi--group theory as also used by, e.g., \cite{Scheinkman1995} to analyze
the properties of continuous--time Markov processes.

Our theoretical results demonstrate that power series expansions of Markov
moments may very well not converge: The chosen moment and model has to
satisfy certain regularity conditions for this to hold. In particular, we
demonstrate that the expansions of transition densities and option prices
proposed by \cite{kristensen2011}, \cite{Yangetal2019} and \cite{Wan2021} do
not converge. That is, these methods are bound to fail as the number of
series terms grows. As such, the expansions proposed in these papers and the
extension to general jump--diffusions developed here should be used with
care. In particular, researchers may not wish to add more than, say, 4--5
terms to the expansion in order to avoid the numerical error to blow up.

The remains of the paper are organized as follows. Section \ref{sec: atd}
presents series expansions of a broad class of moments and densities of
basically any continuous--time Markov process. In section \ref{sec:
Practical}, we propose a numerical implementation of the general method when
applied to general jump--diffusion models. Section \ref{sec: theory}
analyzes the theoretical properties of the power series expansion over both
shrinking and fixed time distances. Section \ref{sec: numeric} examines the
numerical performance of our numerical algorithm. Section \ref{sec:
conclusion} concludes. Appendix \ref{sec: Proofs} gathers all proofs.

\section{Moment expansions of Markov processes\label{sec: atd}}

We first provide a motivating example of a jump--diffusion model and some of
the moments researchers often are interested in computing. We then proceed
to consider more general framework and develop a general moment expansion
method in this setting.

\subsection{Motivating example}

Consider a $d$-dimensional process, $x_{t}\in \mathcal{X\subseteq }\mathbb{R}%
^{d}$ that solves the following stochastic differential equation (SDE): 
\begin{equation}
dx_{t}=\mu \left( x_{t}\right) dt+\sigma \left( x_{t}\right)
dW_{t}+J_{t}dN_{t},  \label{eq: model}
\end{equation}%
where $\mu \left( x\right) $ and $\sigma \left( x\right) $ are the so-called
drift and diffusion functions, respectively, $W_{t}$ is a $d$-dimensional
standard Brownian motion, $N_{t}$ is a Poisson process with jump intensity $%
\lambda \left( x_{t}\right) $, and $J_{t}$ captures the jump-sizes and has
conditional density $\nu \left( \cdot |x_{t}\right) $. The precise form of $%
\mu \left( x\right) $, $\sigma \left( x\right) $, $\lambda \left( x\right) $
and $\nu \left( \cdot |x\right) $ are chosen by the researcher according to
the dynamic problem that is being considered and so are known to us. To keep
notation simple, we restrict ourselves to the time--homogenous case meaning
that none of the functions entering the model depend on $t$; the extension
to the time--inhomogenous case can be found in Appendix \ref{Sec:
Time-inhomo}.

We are interested in computing conditional moments on the form%
\begin{equation}
u_{t}\left( x\right) =E_{t}f\left( x\right)   \label{eq: u FK}
\end{equation}%
where%
\begin{equation}
\left( t,f\right) \mapsto E_{t}f\left( x\right) \equiv \mathbb{E}\left[
\left. \exp \left( -\int\nolimits_{0}^{t}r\left( x_{s}\right) ds\right)
f\left( x_{t}\right) \right\vert x_{0}=x\right]   \label{eq: E_t(f) def}
\end{equation}%
is a conditional moment operator. This family of operators, indexed by the
time variable $t\geq 0$, constitutes a so--called semi--group of linear
operators; for an overview of the general theory of semi--groups with
applications to Markov processes we refer to \cite{ethier1986}; for
applications of semi--group theory in econometrics and finance, see \cite%
{AitSahaliaetal2010}.\footnote{%
Note that we here opt for the so--called Musiela parameterization where $t$
measures the time distance between the current and some future calendar time
point. One could alternatively have defined the function of interest as, for
some given $T<\infty $, 
\begin{equation*}
\tilde{u}_{\tau }\left( x\right) =\mathbb{E}\left[ \left. \exp \left(
-\int\nolimits_{\tau }^{T}r\left( x_{s}\right) ds\right\vert f\left(
x_{T}\right) \right\vert x_{\tau }=x\right] ,
\end{equation*}%
where now $\tau \leq T$ is a calendar time point. In the current
time--homogenous case, it is easily seen that $\tilde{u}_{\tau }\left(
x\right) =u_{T-\tau }\left( x\right) $, where $u_{t}$ was defined in (\ref%
{eq: u FK}).}

The functions $r\left( x\right) $ and $f\left( x\right) $ entering (\ref{eq:
u FK})--(\ref{eq: E_t(f) def}) are chosen by the researcher according to the
problem of interest. For example, with $r\left( x\right) =0$ and $f\left(
x\right) =\delta \left( y-x\right) $ for some fixed $y\in \mathcal{X}$,
where $\delta \left( x\right) $ is Dirac's Delta function, $u_{t}\left(
x\right) =p_{t}\left( y|x\right) $, where $p_{t}$ is the transition density
of $x_{t}$,%
\begin{equation*}
\Pr \left( x_{t}\in \mathcal{A}|x_{0}=x\right) =\int_{\mathcal{A}%
}p_{t}\left( y|x\right) dy,\text{ \ \ }\mathcal{A}\subseteq \mathcal{X}.
\end{equation*}%
If instead we choose $r_{t}\left( x\right) =r>0$ and $f\left( x\right)
=\left( \exp \left( x_{1}\right) -K\right) ^{+}$ then $u_{t}\left( x\right) $
becomes the price of a European call option with time to maturity $t$ when
the state variables $x_{t}$ satisfy (\ref{eq: model}) under the
risk--neutral measure with the first component, $x_{1,t}$, being the
log-price of the underlying asset and the short-term interest rate equals
the constant $r$. When $r\left( x\right) =x_{2}$, $u_{t}\left( x\right) $ is
the price of the same option but now allowing for a stochastic short--term
interest rate, which is the second component of $x_{t}$.

In most cases, an analytic expression of (\ref{eq: u FK}) is not available
and $u_{t}\left( x\right) $ has to be computed using numerical
approximations. To motivate our proposed approximation of $u_{t}\left(
x\right) $, observe that an equivalent representation of it is the solution
to a partial integro-differential equation (PIDE). An important component of
this PIDE is the so--called (infinitesimal) generator $A$ of $x_{t}$ which
fully characterizes the dynamics. The generator is given by, for any
sufficiently regular function $f\left( x\right) $,%
\begin{equation}
Af\left( x\right) =A_{D}f\left( x\right) +A_{J}f\left( x\right) ,
\label{eq: A def}
\end{equation}%
where, with $\sigma ^{2}\left( x\right) :=\sigma \left( x\right) \sigma
\left( x\right) ^{\top }\in \mathbb{R}^{d\times d}$,%
\begin{equation}
A_{D}f\left( x\right) =\sum_{i=1}^{d}\mu _{i}\left( x\right) \partial
_{x_{i}}f\left( x\right) +\frac{1}{2}\sum_{i,j=1}^{d}\sigma
_{ij,t}^{2}\left( x\right) \partial _{x_{i},x_{j}}^{2}f\left( x\right)
\label{eq: A_D def}
\end{equation}%
and%
\begin{equation}
A_{J}f\left( x\right) =\lambda \left( x\right) \int_{\mathbb{R}^{d}}\left[
f\left( x+c\right) -f\left( x\right) \right] \nu \left( c|x\right) dc
\label{eq: A_J def}
\end{equation}%
are the generators of the diffusive and jump component of $x_{t}$,
respectively. Here, $\partial _{x_{i}}f\left( x\right) =\partial f\left(
x\right) /\left( \partial x_{i}\right) $, $\partial
_{x_{i},x_{j}}^{2}f\left( x\right) =\partial ^{2}f\left( x\right) /\left(
\partial x_{i}\partial x_{j}\right) $ and similar for other partial
derivatives.

It can then be shown, c.f. Section \ref{sec: theory}, that $u_{t}\left(
x\right) $ solves the following PIDE:%
\begin{equation}
\partial _{t}u_{t}\left( x\right) =\left[ A-r\left( x\right) \right]
u_{t}\left( x\right) ,\text{ \ \ }t\geq 0,x\in \mathcal{X},
\label{eq: PIDE true}
\end{equation}%
with initial condition $u_{0}\left( x\right) =f\left( x\right) $ for all $%
x\in \mathcal{X}$. In the case of pure diffusions ($A_{J}=0$), the reader
may recognize (\ref{eq: u FK}) as the celebrated Feynman--Kac representation
of the solution to (\ref{eq: PIDE true}) which also holds for the general
case of jump--diffusions. The solution to this PIDE can be represented in
the following abstract manner: $u_{t}\left( x\right) =e^{\left( A-r\right)
t}f\left( x\right) $, where $e^{\left( A-r\right) t}$ is the exponential of
the operator $A-r$ in the sense that%
\begin{equation}
\frac{\partial e^{\left( A-r\right) t}}{\partial t}=\left( A-r\right)
e^{\left( A-r\right) t}.  \label{eq: exp(A-r) deriv}
\end{equation}

We are now interested in obtaining an approximation of $u_{t}\left( x\right) 
$ based on a series expansion w.r.t. time $t$. A simple version of this
would be a Taylor series expansion around $t=0$ on the form 
\begin{equation}
\tilde{u}_{t}\left( x\right) \equiv \sum\limits_{m=0}^{M}\frac{t^{m}}{m!}%
\left. \partial _{t}^{m}u_{t}\left( x\right) \right\vert
_{t=0}=\sum\limits_{m=0}^{M}\frac{t^{m}}{m!}\left( A-r\right) ^{m}f\left(
x\right) ,  \label{eq: u expand naive}
\end{equation}%
for some $M\geq 1$, where the second equality uses (\ref{eq: exp(A-r) deriv}%
). This type of moment approximations have found widespread use in the
literature; see, e.g., 
\citep[see,
e.g.,][]{aitsahalia2002, bakshi2006, yu2007, aitsahalia2008, Filipovic2013, li2013}%
. However, this expansion is not valid (well-defined) when, for example, $%
f\left( x\right) $ is a non--smooth function since the domain of the
operator $A_{D}$ is restricted to smooth functions, c.f. (\ref{eq: A_D def}%
). The transition density and option pricing examples provided above fall in
this category. We will now present a generalized version of above expansion
that circumvents this issue; this is done for fully general
semi--groups/Markov processes.

\subsection{General framework}

We take as given some semi--group\footnote{%
A family of linear operators $\left\{ E_{t}:t\geq 0\right\} $ is said to be
a semi--group if it satisfies (i) $E_{0}f=f$ and and (ii) $%
E_{s+t}f=E_{s}E_{t}f$ for all $s,t\geq 0$..} $\left( t,f\right) \mapsto
E_{t}f\left( x\right) $ of interest. It could, for example, be on the form (%
\ref{eq: E_t(f) def}) for some continuous--time Markov process $x_{t}$, not
necessarily a jump--diffusion process. But we do not restrict ourselves to
this case.

Suppose that $E_{t}f\left( x\right) $ is not available on closed form for a
given choice of $f$; we here show how this can be approximated through a
Taylor series expansion of $E_{t}f\left( x\right) $ w.r.t. $t$ when either $%
f $ is sufficiently regular, where the notion of "regular" will be made
clear below, or it can be expressed as the limit of a regular function. The
proposal is a generalisation of the one of \cite{kristensen2011}, but the
derivation will be carried out using semi--group theory which simplifies the
derivations substantially compared to \cite{kristensen2011} and provides new
insights into the expansion.

Let $\mathcal{D}\left( E\right) $ denote the domain of $E$ and let $B$ be
the infinitesimal operator of the semi--group defined as%
\begin{equation*}
Bf\left( x\right) =\lim_{t\rightarrow 0^{+}}\frac{E_{t}f\left( x\right)
-f\left( x\right) }{t},
\end{equation*}%
and let $\mathcal{D}\left( B\right) $ denote the domain of $B$; that is, the
set of functions $f\in \mathcal{D}\left( E\right) $ for which the above
limit exists. In the motivating jump--diffusion example, $B=A-r$. For any $%
f\in \mathcal{D}\left( E\right) $, we write%
\begin{equation}
E_{t}f\left( x\right) =e^{Bt}f\left( x\right) ,  \label{eq: E_t def gen}
\end{equation}%
where as before this should be interpreted as 
\begin{equation}
\partial _{t}E_{t}f\left( x\right) =BE_{t}f\left( x\right) ,\text{ \ \ }t>0.
\label{eq: PIDE general}
\end{equation}%
If the chosen $f$ satisfies $f\in \mathcal{D}\left( B^{M}\right) $, (\ref%
{eq: PIDE general}) also holds at $t=0$ and so the following Taylor series
expansion of $E_{t}f\left( x\right) $ w.r.t. $t$ around $t=0$ is valid, 
\begin{equation}
\hat{E}_{t}f\left( x\right) \equiv \sum\limits_{m=0}^{M}\frac{t^{m}}{m!}%
\left. \partial _{t}^{m}E_{t}f\left( x\right) \right\vert
_{t=0^{+}}=\sum\limits_{m=0}^{M}\frac{t^{m}}{m!}B^{m}f\left( x\right) ,
\label{eq: E approx regular}
\end{equation}%
where under weak conditions $\tilde{E}_{t}f\left( x\right) =E_{t}f\left(
x\right) +O\left( t^{M}\right) $. This is a generalised version of (\ref{eq:
u expand naive}).

We are now interested in generalising this series expansion to also work
when $f\notin \mathcal{D}\left( B^{M}\right) $. An important ingredient of
this is to first identify/construct a smoothed version of $f\left( x\right) $%
, denoted $u_{0,s}\left( x\right) $, $s\geq 0$ and $x\in \mathcal{X}$, which
we require to satisfy the following two conditions:

\begin{description}
\item[A.0] (i) $\lim_{s\rightarrow 0^{+}}u_{0,s}\left( x\right) =f\left(
x\right) $ and (ii) $u_{0,s}\in \mathcal{D}\left( \left( \partial
_{s}\right) ^{M_{1}}\right) \cap \mathcal{D}\left( B^{M_{2}}\right) $ for
some $s\geq 0$ and $M_{1},M_{2}\geq 1$.
\end{description}

The function $u_{0,s}\left( x\right) $ is chosen by the researcher and needs
to be available on closed form for the subsequent approximation to be
operational. The choice of $u_{0,s}\left( x\right) $ is application specific
in the sense that Assumption A.0 has to be satisfied: Part (i) requires $%
u_{0,s}\left( x\right) $ to converge towards the irregular function of
interest $f\left( x\right) \notin \mathcal{D}\left( B\right) $ as $%
s\rightarrow 0^{+}$. Part (ii) says that, for some $s>0$, $u_{0,s}\left(
x\right) $ is sufficiently regular in the sense that it is $M$ times
continuously differentiable in $s$ and each of these derivatives belongs to $%
\mathcal{D}\left( B^{M}\right) $.

Assumption A.0 allows for a broad range of smoothers. One choice of $%
u_{0,s}\left( x\right) $ which under great generality will satisfy A.0 is $%
u_{0,s}\left( x\right) =E_{0,s}f\left( x\right) $ where $E_{0}$ is another
semi--group chosen such that $u_{0,s}\left( x\right) $ is available on
closed form. This choice clearly satisfies part (i) and if $E_{0}$ has
similar properties as the one of interest, $E$, so that their respective
generators have shared domain, then part (ii) will also hold. A simple
choice of $E_{0}$, as proposed by \cite{kristensen2011}, is $E_{0,s}f\left(
x\right) =E\left[ f\left( x_{0,s}\right) |x_{0}=x\right] ,$where $x_{0,s}$
is another stochastic process specified by the researcher. The process $%
x_{0,t}$ could, for example, be chosen as a random walk type stochastic
process with transition density $p_{0,s}\left( y|x\right) =K\left( \left(
y-x\right) /\sqrt{s}\right) /\sqrt{s}$, for some kernel density $K:\mathbb{R}%
^{d}\mapsto \mathbb{R}^{d}$, in which case%
\begin{equation}
u_{0,s}\left( x\right) =\frac{1}{\sqrt{s}}\int_{\mathbb{R}^{d}}f\left(
y\right) K\left( \frac{y-x}{\sqrt{s}}\right) dy.  \label{eq: u kernel reg}
\end{equation}%
This choice satisfies (i) and if $K$ is $M_{1}$ times differentiable then $%
u_{0,s}\left( x\right) $ has the same property. The final requirement, $%
u_{0,s}\in \mathcal{D}\left( B^{M_{2}}\right) $, has to be checked on a case
by case basis.

Under A.0, the following identity holds:%
\begin{equation*}
f\left( x\right) =u_{0,0}\left( x\right) =e^{\left( -\partial _{s}\right)
s}u_{0,s}\left( x\right) ,
\end{equation*}%
where the second equality simply states that $u_{0,0}\left( x\right)
=u_{0,s}\left( x\right) +\int_{0}^{s}\left( -\partial _{\tau }\right)
u_{0,\tau }\left( x\right) d\tau $. Substituting this into (\ref{eq: E_t def
gen}) yields%
\begin{equation}
E_{t}f\left( x\right) =e^{Bt}e^{\left( -\partial _{s}\right) s}u_{0,s}\left(
x\right) =e^{\left( -\partial _{s}\right) s}e^{Bt}u_{0,s}\left( x\right) ,
\label{eq: u_t id}
\end{equation}%
where the last equality uses the following fundamental result: If two
infinitesimal operators, say, $B_{1}$ and $B_{2}$, commute in the sense that 
$B_{1}B_{2}f=B_{2}B_{1}f$ then $%
e^{B_{1}s}e^{B_{2}t}f=e^{B_{1}s+B_{2}t}f=e^{B_{2}t}e^{B_{1}s}f$. This
applies to the case of $B$ and $\partial _{s}$, $\partial _{s}B=B\partial
_{s}$, since $B$ acts on $x$ while $\partial _{s}$ acts on $s$.

Finally, carry out a Taylor series expansion w.r.t. $\left( s,t\right) $ to
obtain 
\begin{equation}
\hat{E}_{t}f\left( x\right)
=\sum\limits_{m_{1}=0}^{M_{1}}\sum\limits_{m_{2}=0}^{M_{2}}\frac{\left(
-s\right) ^{m_{1}}t^{m_{2}}}{m_{1}!m_{2}!}B^{m_{2}}\partial
_{s}^{m_{1}}u_{0,s}\left( x\right) ,  \label{eq: E approx general}
\end{equation}%
where the order of $\partial _{s}$ and $B$ can be exchanged since $\partial
_{s}^{m_{1}}B^{m_{2}}u_{0,s}\left( x\right) =B^{m_{2}}\partial
_{s}^{m_{1}}u_{0,s}\left( x\right) $. The resulting approximation error is
of order $O\left( s^{M_{1}}\right) +O\left( t^{M_{2}}\right) $. In
particular, the above expansion will generally be more precise as $s$ gets
smaller. Thus, we ideally want to choose $s$ as small as possible to reduce
the approximation error. However, for the chosen value of $s\geq 0$ A.0(ii)
has to be satisfied. This rules out, for example, $s=0$ when $f$ is
irregular since $u_{0,0}\left( x\right) =f\left( x\right) $.

However, if the approximation error is not a major concern (which is, for
example, the case if the order of approximation can be chosen sufficiently
large) then one can choose $s=t$ in which case $E_{t}f\left( x\right)
=e^{\left( B-\partial _{t}\right) t}u_{0,t}\left( x\right) $ and the
following special case of (\ref{eq: E approx general}) can be employed,%
\begin{equation}
\hat{E}_{t}f\left( x\right) =\sum\limits_{m=0}^{M}\frac{t^{m}}{m!}\left(
B-\partial _{t}\right) ^{m}u_{0,t}\left( x\right) .
\label{eq: E approx special}
\end{equation}

We show in Appendix \ref{sec: relationship} that (\ref{eq: E approx special}%
) is a generalized version of the proposal of \cite{kristensen2011} which in
turn contains as special cases the expansions of \cite{Yangetal2019} and 
\cite{Wan2021}.

\section{Implementation of expansion for jump-diffusion models\label{sec:
Practical}}

This section provides details regarding the practical implementation of the
proposed approximation in the jump--diffusion case. We here focus on the
special case of $r\left( x\right) =0$ and $s=t$, in which case $u_{t}\left(
x\right) =E_{t}f\left( x\right) =\mathbb{E}\left[ f\left( x_{t}\right)
|x_{0}=x\right] $ and%
\begin{equation}
\hat{u}_{t}\left( x\right) =\sum\limits_{m=0}^{M}\frac{t^{m}}{m!}\left(
A-\partial _{t}\right) ^{m}u_{0,t}\left( x\right) .
\label{eq: u-hat special}
\end{equation}%
This is done to avoid overly complicated notation. Most of the ideas and
arguments extend to the general case.

\subsection{Choice of smoothing function for irregular moments}

Following \cite{kristensen2011}, a simple choice of $u_{0,s}\left( x\right) $
that satisfies A.1 is $u_{0,s}\left( x\right) =E_{0,s}f\left( x\right) =%
\mathbb{E}\left[ f\left( x_{0,s}\right) |x_{0,0}=x\right] $ where $x_{0,s}$
is chosen as the solution to an auxiliary jump--diffusion model,%
\begin{equation}
dx_{0,t}=\mu _{0}\left( x_{0,t}\right) dt+\sigma _{0}\left( x_{0,t}\right)
dW_{t}+J_{0,t}dN_{0,t},
\end{equation}%
where $N_{0,t}$ is a Poisson process with jump intensity $\lambda _{0}\left(
x\right) $ and $J_{0,t}$ has density $\nu _{0}\left( \cdot |x\right) $. The
auxiliary model should be chosen so that $u_{0,t}\left( x\right) $ is
available on closed form. One such model is the multivariate Brownian motion
with drift model,%
\begin{equation}
dx_{0,t}=\mu _{0}dt+\sigma _{0}dW_{t}  \label{eq: BM aux}
\end{equation}%
where $\mu _{0}\in \mathbb{R}^{d}$ and $\sigma _{0}\in \mathbb{R}^{d\times
d} $ are constants, or the multivariate Vasicek (Ornstein--Uhlenbeck) model,%
\begin{equation*}
dx_{0,t}=\left( \mu _{0}+Ax_{0t}\right) dt+\sigma _{0}dW_{t},
\end{equation*}%
both of which have a Gaussian transition density on known form. In either
case,%
\begin{equation*}
u_{0,s}\left( x\right) =\int f\left( y\right) p_{0,s}\left( y|x\right) dy,
\end{equation*}%
where $p_{0,s}\left( y|x\right) $ is the transition density of the auxiliary
model. For example, in the case of (\ref{eq: BM aux}),%
\begin{equation}
p_{0,t}\left( y|x\right) =\frac{1}{\sqrt{2\pi t\left\vert \sigma
_{0}^{2}\right\vert }}\exp \left( -\frac{\left( y-x-t\mu _{0}\right)
^{\prime }\sigma _{0}^{-2}\left( y-x-t\mu _{0}\right) }{2t}\right) .
\label{eq: p_0 BM}
\end{equation}%
Note that with $\mu _{0}=0$ above specification corresponds to (\ref{eq: u
kernel reg}) with $K$ chosen as the Gaussian kernel.

Recall the two motivating examples of transition density and option price
approximation. In the case of $f\left( x\right) =\delta \left( y-x\right) $,
we get $u_{0,s}\left( x\right) =p_{0,s}\left( y|x\right) $. If $f\left(
x\right) =\left( \exp \left( x_{1}\right) -K\right) ^{+}$, and we set $\mu
_{0,1}=r-\sigma _{0,11}^{2}/2$ to ensure risk--neutrality in the auxiliary
model, then $u_{0,s}\left( x\right) $ takes the form of the well-known
formula for the risk--neutral expected pay-off of a call option in the
Black--Scholes model, 
\begin{equation}
u_{0,s}\left( x\right) =se^{rs}\Phi \left( d_{+}\left( x,s\right) \right)
-K\Phi \left( d_{-}\left( x,s\right) \right) ,  \label{eq: B-S pay-off}
\end{equation}%
where $d_{\pm }\left( x,s\right) =\left( x-\log \left( K\right) +\left( r\pm 
\frac{1}{2}\sigma _{0,11}^{2}\right) s\right) /\left( \sigma _{0,11}\sqrt{s}%
\right) $ and $\Phi \left( \cdot \right) $ denotes the cdf of the $N\left(
0,1\right) $ distribution.

\subsection{Pure diffusion case}

In the pure diffusion case, where no jump component is present so that $%
A_{J}=0$, analytical expressions of $\left( A_{D}-\partial _{t}\right)
^{m}u_{0,t}\left( x\right) $ are in principal straightforward to obtain
relying on symbolic software packages, such as Mathematica, since $A_{D}$ is
a differential operator. We refer to \cite{kristensen2011}, \cite%
{Yangetal2019} and \cite{Wan2021} for more details on this for the two
leading examples of density and option price approximations and with $u_{0,t}
$ chosen as the corresponding solution under (\ref{eq: BM aux}).

\subsection{Jump-diffusion case}

\subsubsection{State--independent jump or diffusion component}

Next, consider jump--diffusion models where either the diffusive component
or the jump component of $x_{t}$ are state--independent; the latter case
corresponds to the class of jump--diffusions considered in \cite{Wan2021}.

These two cases correspond to (i) $\mu \left( x\right) =\mu $ and $\sigma
^{2}\left( x\right) =\sigma ^{2}$ are constant or (ii) $\lambda \left(
x\right) =\lambda $ and $\nu \left( \cdot |x\right) =\nu \left( \cdot
\right) $ are independent of $x$, respectively. In either case, we can write 
$x_{t}=x_{D,t}+x_{J,t}$ where the diffusive component, $x_{D,t}$, and the
jump component, $x_{J,t}$, are now mutually independent. As a consequence,
the two generators $A_{D}$ and $A_{J}$ commute, $A_{D}A_{J}=A_{J}A_{D}$, in
which case%
\begin{equation}
u_{t}\left( x\right) =e^{\left( A_{D}+A_{J}\right) t}f\left( x\right)
=e^{A_{D}t}B_{J,t}\left( x\right) =e^{A_{J}t}B_{D,t}\left( x\right) ,
\label{eq: w A B commute}
\end{equation}%
where 
\begin{equation*}
B_{J,t}\left( x\right) =\mathbb{E}\left[ \left. f\left( x_{J,t}\right)
\right\vert x_{J,0}=x\right] ,\text{ \ \ }B_{D,t}\left( x\right) =\mathbb{E}%
\left[ \left. f\left( x_{D,t}\right) \right\vert x_{D,0}=x\right] .
\end{equation*}

Now, consider first the case where (ii) is satisfied. In this scenario, $%
x_{J,t}|x_{J,0}=x$ has density 
\begin{equation}
p_{J,t}\left( y|x\right) =\sum_{k=0}^{\infty }e^{-\lambda t}\frac{\left(
\lambda t\right) ^{k}}{k!}\nu _{k}\left( y-x\right) ,  \label{eq: p_J expr}
\end{equation}%
where $\nu _{k}\left( y\right) $ is the density of the sum of $k$
independent jumps, $\sum_{i=1}^{k}J_{i}$, $J_{i}\sim \nu \left( \cdot
\right) $. Since $p_{J,t}\left( y|x\right) $ is a smooth function then $%
x\mapsto B_{J,t}\left( \Delta x_{D,T}+x\right) $ is also a smooth function
even if $f\left( x\right) $ is irregular. Thus, if $B_{J,t}\left( x\right) $
is available on closed form then the smoothing device is not needed and we
can approximate $u_{t}$ by%
\begin{equation}
\hat{u}_{t}\left( x\right) =\sum\limits_{m=0}^{M}\frac{t^{m}}{m!}%
A_{D}^{m}B_{J,t}\left( x\right) .  \label{eq: w-hat J ver}
\end{equation}

Similar, if (i) is satisfied then $x_{D,t}|x_{D,0}$ is a Brownian motion
with drift and has Gaussian density as given in (\ref{eq: p_0 BM}). Because
of its simple dynamics, $B_{D,t}\left( x\right) $ is available on closed
form in many cases and will again be a smooth function; if so, we propose to
approximate $u_{t}$ by%
\begin{equation*}
\hat{u}_{t}\left( x\right) =\sum\limits_{m=0}^{M}\frac{t^{m}}{m!}%
A_{J}^{m}B_{D,t}\left( x\right) .
\end{equation*}%
If closed form expressions of neither $B_{D,t}$ nor $B_{J,t}$ are available,
it is still possible to simplify the computation using, for example,%
\begin{equation}
\hat{u}_{t}\left( x\right) =\left[ \sum\limits_{m=0}^{M}\frac{t^{m}}{m!}%
e^{A_{J}t}\left( A_{D}-\partial _{t}\right) ^{m}u_{0,t}\left( x\right) %
\right] ,  \label{eq: w-hat J ver 2}
\end{equation}%
assuming that closed form expressions of $e^{A_{J}\left( T-t\right) }\left(
A_{D}-\partial _{t}\right) ^{m}u_{0,t}\left( x\right) $ can be computed.
This last version is the one proposed by \cite{Wan2021} for jump--diffusions
with state--independent jumps.

\subsubsection{State--dependent jump and diffusive component}

Finally, consider the general case where $A_{J}\neq 0$ and both the
diffusion and jump component are state--dependent. First observe that when
the jumps are state--dependent, or have a complex distribution, $%
A_{J}f\left( x\right) $ cannot be evaluated analytically for a given
function $f$ in general. We propose to resolve this issue by approximating
the integral part of $A_{J}f\left( x\right) $, $A_{J1}f\left( x\right)
=\lambda \left( x\right) \int_{\mathbb{R}^{d}}f\left( x+c\right) \nu \left(
c\right) dc$, by%
\begin{equation}
\hat{A}_{J1}f\left( x\right) =\lambda \left( x\right) \sum_{s=1}^{S}\omega
_{s}f\left( x+c_{s}\right) ,  \label{eq: L-hat-J}
\end{equation}%
where $\omega _{s}$ and $c_{s}$, $s=1,...,S$, are integration weights and
nodes, respectively. For example, in the case of Monte Carlo integration
with $S$ random draws from $\nu $, $\omega _{s}=1/S$ and $c_{s}$ is the $s$%
th draw from $\nu \left( \cdot \right) $. The resulting approximate operator 
$\hat{A}_{J}f\left( x\right) =\hat{A}_{J1}f\left( x\right) -\lambda \left(
x\right) f\left( x\right) $ is on closed form and so we can now continue as
in the pure diffusion case. Also note that $\hat{A}_{J1}f\left( x\right)
\rightarrow A_{J1}f\left( x\right) $ as $S\rightarrow \infty $ which ensures
that the added numerical error can be controlled by choosing $S$ large
enough.

In the case that $v\left( c\right) $ belongs to the exponential family, the
generator of jump component, $A_{J1}$, is well--approximated using
Gauss-Hermite or Gauss-Laguerre quadrature. For example, when $J_{t}$ is
i.i.d. scalar with double exponential distribution with mean zero and
standard deviation $\sigma _{J}$, it follows from a change of variables that 
\begin{eqnarray*}
\int_{-\infty }^{\infty }f\left( x+c\right) \frac{1}{2\sigma _{J}}e^{-\frac{%
\left\vert c\right\vert }{\sigma _{J}}}dc &=&\int_{0}^{\infty }\left[
f\left( x+c\right) +f\left( x-c\right) \right] \frac{1}{2\sigma _{J}}e^{-%
\frac{c}{\sigma _{J}}}dc \\
&=&\frac{1}{2}\int_{0}^{\infty }\left[ f\left( x+\sigma _{J}c\right)
+f\left( x-\sigma _{J}c\right) \right] e^{-c}dc.
\end{eqnarray*}%
Then, given the nodes and weights, $c_{s}^{GL}$ and $\omega _{s}^{GL}$, for
the Gauss-Laguerre quadrature, the approximation takes the following form: 
\begin{equation*}
\int_{-\infty }^{\infty }f\left( x+c\right) \frac{1}{2\sigma _{J}}e^{-\frac{%
\left\vert c\right\vert }{\sigma _{J}}}dc\simeq \frac{1}{2}%
\sum_{s=0}^{S_{GL}-1}w_{s}^{GL}\left[ f\left( x+\sigma _{J}c_{s}^{GL}\right)
+f\left( x-\sigma _{J}c_{s}^{GL}\right) \right] .
\end{equation*}%
We use this approximation method in our numerical studies when we cannot
obtain an exact expression of the integral (as discussed with standard
packages such as Mathematica, or it may not be evaluated through equally
standard packages such as Matlab. We find that Gaussian quadrature is more
accurate and easier to implement than Monte Carlo methods with low
computational cost.

With $\hat{A}_{J1}$ replacing $A_{J1}$, we can now use a symbolic software
package to obtain expressions of $\left( A_{D}+\hat{A}_{J}-\partial
_{t}\right) ^{m}u_{0,t}$, $m=1,2,\ldots $. For example, 
\begin{equation*}
\left( A_{D}+\hat{A}_{J}-\partial _{t}\right) ^{2}u_{0,t}=\left(
A_{D}-\lambda -\partial _{t}\right) ^{2}u_{0,t}+\hat{A}_{J1}\left(
A_{D}-\lambda -\partial _{t}\right) u_{0,t}+\left( A_{D,t}-\lambda -\partial
_{t}\right) \hat{A}_{J1}u_{0,t}+\hat{A}_{J1}^{2}u_{0,t},
\end{equation*}%
where the evaluation of $\left( A_{D}-\lambda -\partial _{t}\right) ^{2}f$
and $\left( A_{D}-\lambda -\partial _{t}\right) f$ can done using symbolic
methods while (here in the univariate case for simplicity) 
\begin{eqnarray*}
\left( A_{D}-\lambda -\partial _{t}\right) \hat{A}_{J1}u_{0,t} &=&\partial
_{t}\left\{ \lambda \left( x\right) \sum_{s=1}^{S}\omega _{s}u_{0,t}\left(
x+c_{s}\right) \right\} +\mu \left( x\right) \partial _{x}\left\{ \lambda
\left( x\right) \sum_{s=1}^{S}\omega _{s}u_{0,t}\left( x+c_{s}\right)
\right\}  \\
&&+\frac{1}{2}\sigma ^{2}\left( x\right) \partial _{x}^{2}\left\{ \lambda
\left( x\right) \sum_{s=1}^{S}\omega _{s}u_{0,t}\left( x+c_{s}\right)
\right\} -\lambda ^{2}\left( x\right) \sum_{s=1}^{S}\omega _{s}u_{0,t}\left(
x+c_{s}\right) ,
\end{eqnarray*}%
and 
\begin{eqnarray*}
\hat{A}_{J1}^{2}u_{0,t}\left( x\right)  &=&\lambda \left( x\right)
\sum_{s_{1}=1}^{S}\omega _{s_{1}}\hat{A}_{J1}u_{0,t}\left(
x+c_{s_{1}}\right)  \\
&=&\lambda \left( x\right) \sum_{s_{1}=1}^{S}\omega _{s_{1}}\left[ \lambda
\left( x+c_{s_{1}}\right) \sum_{s_{2}=1}^{S}\omega _{s2}u_{0,t}\left(
x+c_{s_{1}}+c_{s_{2}}\right) \right] .
\end{eqnarray*}

\section{Theoretical properties\label{sec: theory}}

We first present a general theory of series expansions on the form (\ref{eq:
E approx regular}) when the function $f$ is regular in the sense that $f\in
D\left( B\right) $. We provide two sets of results: First, we derive an
error bound for any given value $M$ of the order of the expansion. Second,
we provide conditions under which the error bound vanishes as $M\rightarrow
\infty $ at a given value of the time horizon $t>0$. The conditions for the
second set of results come in two forms: We first provide conditions under
which the proposed power series expansion converges globally, i.e., over the
whole domain of $x_{t}$. These conditions are somewhat restrictive though
and rule out certain models and functions of interest. We therefore proceed
to examine how the approximation behaves on a given compact subset of the
full domain, and show that the power series expansion is consistent over
compact subsets under weak regularity conditions that most known models
satisfy. We then apply the theory to moments of jump--diffusions on the form
(\ref{eq: model}) and provide primitive conditions under which the expansion
is valid. Some of the results presented here rely on the important insights
found in the unpublished work of \cite{schaumburg2004} which we are indebted
to.

Next, we then proceed to analyze the "smoothed" expansion (\ref{eq: E approx
general}). As in the regular case, we are able to derive an error bound for
a given choice of $M$. But at the same time, this expansion is generally not
consistent in the sense that it will not converge as $M\rightarrow \infty $
for a fixed value of $t>0$. This is an important result since this shows
that the approximation error will eventually blow up as we increase $M$.
Thus, researchers should use the generalized version with caution.

\subsection{Series expansions of regular "moments"}

We take as given a semi--group $E_{t}:\mathcal{F\mapsto F}$ where $\mathcal{F%
}$  is equipped with some function norm $\left\Vert \cdot \right\Vert _{%
\mathcal{F}}$. In the leading case of $E_{t}f\left( x\right) =\mathbb{E}%
\left[ f\left( x_{t}\right) |x_{0}=x\right] $, two standard choices of $%
\left( \mathcal{F},\left\Vert \cdot \right\Vert _{\mathcal{F}}\right) $ are
the following: The first is the space of bounded functions equipped with a $%
\sup $ norm, $\left\Vert f\right\Vert _{\mathcal{F}}=\sup_{x\in \mathcal{X}%
}\left\vert f\left( x\right) \right\vert $. The second is the space of
functions with second moments equipped with the following $L_{2}$ norm,$\
\left\Vert f\right\Vert _{\mathcal{F}}^{2}=\int_{\mathcal{X}}^{\infty
}f^{2}\left( x\right) \pi \left( x\right) dx$ for some weighting function $%
\pi \left( x\right) $. In case of $x_{t}$ being stationary, a natural choice
for $\pi $ is the stationary marginal distribution in which case $\left\Vert
f\right\Vert _{\mathcal{F}}^{2}=\mathbb{E}\left[ f^{2}\left( x_{t}\right) %
\right] $; this norm was, for example, used by \cite{Scheinkman1995}.

We now formally introduce the so--called generator associated with $E_{t}$.
We will here work with the so--called extended generator which is defined as
follows (see, e.g., \cite{Meyn1993}):

\begin{definition}
\label{Def generator}We denote by $\mathcal{D}\left( B\right) $ the set of
functions $f\in \mathcal{F}$ for which there exists $g\in \mathcal{F}$ such
that, for each $t\geq 0$,%
\begin{equation}
E_{t}f=f+\int_{0}^{t}E_{s}gds,\text{ \ \ }\left\Vert E_{t}\left\vert
g\right\vert \right\Vert _{\mathcal{F}}<\infty ,  \label{eq: A extend def}
\end{equation}%
and we write $Bf:=g$ and call $B$ the (extended) generator of $E_{t}$.
\end{definition}

For a given function $f\in \mathcal{F}$, we will in the following frequently
use $u_{t}\left( x\right) $ to denote%
\begin{equation}
u_{t}\left( x\right) \equiv E_{t}f\left( x\right)  \label{eq: w def}
\end{equation}%
to economize on notation. As a first step, we show that $u_{t}\left(
x\right) $ solves (\ref{eq: PIDE general}) if $f\in \mathcal{D}\left(
B\right) $:

\begin{equation}
u_{t}\left( x\right) =f\left( x\right) +\int_{0}^{t}E_{s}\left( Bf\right)
\left( x\right) ds=f\left( x\right) +\int_{0}^{t}Bu_{s}\left( x\right) ds,
\label{eq: w representation}
\end{equation}%
or, equivalently,%
\begin{equation}
\partial _{t}u_{t}\left( x\right) =Bu_{t}\left( x\right) ,\text{ \ \ }t>0,%
\text{ \ \ }u_{0}\left( x\right) =f\left( x\right) .  \label{eq: u PIDE}
\end{equation}

\begin{theorem}
\label{Th: w representation}For any $f\in \mathcal{D}\left( B\right) $, $%
u_{t}\left( x\right) $ in (\ref{eq: w def}) satisfies:

\begin{enumerate}
\item For any fixed $t\geq 0$, $x\mapsto u_{t}\left( x\right) \in \mathcal{D}%
\left( B\right) $ with $Bu_{t}\left( x\right) =E_{t}\left( Bf\right) \left(
x\right) $.

\item If $t\mapsto E_{t}\left( Bf\right) $ is right-continuous at $t=0^{+}$
then $u_{t}\left( x\right) $ solves (\ref{eq: u PIDE}).
\end{enumerate}
\end{theorem}

The continuity condition in the second part of the theorem is satisfied
under great generality when $E_{t}$ is on the form (\ref{eq: E_t(f) def}). A
sufficient condition is that the mapping $\left\{ x_{t}:t\geq 0\right\} $ is
Borel measurable w.r.t. the product sigma algebra, c.f. p. 771 in \cite%
{Scheinkman1995}. The above result, and many subsequent ones, requires the
function $f$ defining $u_{t}\left( x\right) $ to satisfy $f\in \mathcal{D}%
\left( B\right) $. Unfortunately, it rarely easy to give an explicit
characterization of $\mathcal{D}\left( B\right) $. Instead, we will often
work in a smaller subspace, say, $\mathcal{D}_{0}\left( B\right) \subseteq 
\mathcal{D}\left( B\right) $ which is known to us; see Section \ref{sec:
conv jump-diffusion} for an example. One says that $\mathcal{D}_{0}\left(
B\right) $ is a core of $\mathcal{D}\left( B\right) $ if it is a dense
subset of the latter.

We recognize (\ref{eq: u PIDE}) as a generalized version of the celebrated
Kolmogorov's backward equation for jump-diffusion models. In particular, it
implies that $\lim_{t\rightarrow 0^{+}}\partial _{t}u_{t}\left( x\right)
=Bf\left( x\right) $. More generally, under suitable regularity conditions, $%
t\mapsto u_{t}\left( x\right) $ will be $M\geq 1$ times differentiable with%
\begin{equation}
\lim_{t\rightarrow 0^{+}}\partial _{t}^{m}u_{t}\left( x\right) =B^{m}f\left(
x\right) ,\text{ \ \ }0\leq m\leq M,  \label{eq: w diff}
\end{equation}%
in which case the following Taylor series approximation is valid,%
\begin{equation}
\hat{u}_{t}\left( x\right) :=\sum_{m=0}^{M}\frac{t^{m}}{m!}B^{m}f\left(
x\right) .  \label{eq: w-hat def}
\end{equation}%
In order for $\hat{u}_{t}\rightarrow u_{t}$ as $M\rightarrow \infty $, we
need $t\mapsto u_{t}$ to be analytic:

\begin{definition}
\label{Def w analytic}$t\mapsto u_{t}$ is said to be analytic (at $t=0^{+}$)
with radius $T_{0}>0$ if it is infinitely differentiable w.r.t. $t$ and
satisfies 
\begin{equation}
u_{t}=\lim_{M\rightarrow \infty }\hat{u}_{t},\text{ \ }t\leq T_{0}.
\label{eq: power series representation}
\end{equation}
\end{definition}

The definition of $B$ and the convergence result (\ref{eq: power series
representation}) are stated w.r.t. the chosen function norm $\left\Vert
\cdot \right\Vert _{\mathcal{F}}$ introduced earlier. As we shall see,
different assumptions regarding the model and the chosen function $f$
defining $u$ motivate different spaces and norms. Ideally, we would like the
convergence to take place uniformly over all values of $x\in \mathcal{X}$,
but this will only hold for a small set of functions $f$ and models, and so
in some applications it is necessary to work with the weaker $L_{2}$ norm.

In order for $u_{t}$ to be analytic, we need as a minimum that $u_{t}$ is
infinitely differentiable so that (\ref{eq: w diff}) holds for all $m\geq 1$%
. This in turn requires $B^{m}f\left( x\right) $, $m\geq 1$, to be
well--defined. That is, $f\in \mathcal{D}\left( B^{m}\right) $, $m\geq 1$,
where the domains are defined recursively as%
\begin{equation*}
\mathcal{D}\left( B^{m}\right) =\left\{ f\in \mathcal{D}\left(
B^{m-1}\right) :Bf\in \mathcal{D}\left( B\right) \right\} \subseteq \mathcal{%
D}\left( B^{m-1}\right) ,\text{ \ \ }m=2,3,...
\end{equation*}%
The following result shows that the Taylor series $\hat{u}_{t}\left(
x\right) $ is a valid approximation for any $f\in \mathcal{D}\left(
B^{M+1}\right) $ and also provide an error bound for it:

\begin{theorem}
\label{Th: w-hat error bound}For any $f\in \mathcal{D}\left( B^{M+1}\right) $
and $t\geq 0$, $\hat{u}_{t}\left( x\right) $ in (\ref{eq: w-hat def})
satisfies%
\begin{eqnarray*}
\left\vert u_{t}\left( x\right) -\hat{u}_{t}\left( x\right) \right\vert
&=&\left\vert \int_{0}^{t}\int_{0}^{t_{1}}\cdots
\int_{0}^{t_{M}}B^{M+1}u_{t_{M+1}}\left( x\right) dt_{M+1}\cdots
dt_{1}\right\vert \\
&\leq &\frac{t^{M+1}}{\left( M+1\right) !}\sup_{0\leq s\leq t}\left\vert
B^{M+1}u_{s}\left( x\right) \right\vert .
\end{eqnarray*}
\end{theorem}

We recognize the error bound as a generalized version of the one that holds
for a Taylor series approximation of a $M+1$ times differentiable function.
The error bound can be used to show convergence of our expansion of the
transition density with $M\geq 1$ fixed as the time distance between
observations, corresponding to $t$, shrinks to zero. This is the standard
result found in the existing literature on expansions of moments of
continuous-time processes. But, based on this result alone, the
corresponding approximate moment is then only guaranteed to converge towards
the exact one when high-frequency data is available. That is, when $t$
shrinks to zero as the number of observations diverge. For a fixed $t$,
there is no reason why the error bound provided in the theorem will not blow
up as $M\rightarrow \infty $.

We will therefore now derive conditions that guarantee convergence for a
given fixed $t>0$. From Theorem \ref{Th: w-hat error bound} we see that
convergence of $\hat{u}_{t}\left( x\right) $ requires the following two
conditions to be satisfied: $f\in \mathcal{D}\left( B^{\infty }\right) $ and 
$\left\Vert \frac{t^{m}}{m!}B^{m}f\right\Vert _{\mathcal{F}}\rightarrow 0$
as $m\rightarrow \infty $. The convergence result will generally not hold
for all $t>0$. Formally, the radius of convergence is given by 
\begin{equation}
T_{0}=1/\lim \sup_{m\rightarrow \infty }\left\{ \left\Vert B^{m}f\right\Vert
_{\mathcal{F}}/m!\right\} ^{1/m}.  \label{eq: T_0 def}
\end{equation}%
Often the exact value of $T_{0}$ cannot be derived, but it may still be
possible to identify a lower bound for it. Similarly, it is in many
applications difficult to provide a precise characterization of $\mathcal{D}\left( B^{\infty }\right) =\bigcap\nolimits_{m=1}^{\infty }\mathcal{D}\left( B^{m}\right) $. One partial characterization is that it constitutes a
core of $\mathcal{D}\left( B\right) $, c.f. Theorem 7.4.1 of \cite{davies2007}, so that most functions in $\mathcal{D}\left( B\right) $ also belongs to $\mathcal{D}\left( B^{\infty }\right) $. But this provides no
guarantees for that a given function in $\mathcal{D}\left( B\right) $
belongs to $\mathcal{D}\left( B^{\infty }\right) $.

Instead one may seek to identify a subset $\mathcal{F}_{0}\subseteq \mathcal{%
F}$ so that (i) $\mathcal{F}_{0}\subseteq \mathcal{D}\left( B\right) $ and
(ii) the image $B\left( \mathcal{F}_{0}\right) =\left\{ Bf|f\in \mathcal{F}%
_{0}\right\} \subseteq \mathcal{F}_{0}$. For a given $f\in \mathcal{F}_{0}$,
part (i) ensures that $Bf$ is well-defined while part (ii) implies that $%
Bf\in \mathcal{F}_{0}$. In particular, (i)--(ii) guarantee that $\mathcal{F}%
_{0}\subseteq \mathcal{D}\left( B^{m}\right) $ for all $m\geq 1$. As a
consequence, $\mathcal{F}_{0}\subseteq \mathcal{D}\left( B^{\infty }\right) $
thereby providing us with a partial characterization of $\mathcal{D}\left(
B^{\infty }\right) $. In particular, for any given $f\in \mathcal{F}_{0}$,
we have that $t\mapsto u_{t}$ is infinitely differentiable. The following
theorem states the formal result of the above analysis:

\begin{theorem}
\label{Th: analytic}Suppose that $f\in \mathcal{D}\left( B^{\infty }\right) $%
. Then $u_{t}$ is infinitely differentiable and, with the radius of
convergence $T_{0}\geq 0$ given in (\ref{eq: T_0 def}), 
\begin{equation*}
\forall t\leq T_{0}:\left\Vert u_{t}-\hat{u}_{t}\right\Vert _{\mathcal{F}%
}\leq \frac{\left( t/T_{0}\right) ^{M+1}}{1-t/T_{0}}\rightarrow 0\text{ as }%
M\rightarrow \infty .
\end{equation*}

The domain $\mathcal{D}\left( B^{\infty }\right) $ is a core of $\mathcal{D}%
\left( B\right) $. A sufficient condition for $f\in \mathcal{D}\left(
B^{\infty }\right) $, is that $f\in \mathcal{F}_{0}$ for some $\mathcal{F}%
_{0}\subseteq \mathcal{D}\left( B\right) $ satisfying $B\left( \mathcal{F}%
_{0}\right) \subseteq \mathcal{F}_{0}$.
\end{theorem}

The last part of the theorem provides one sufficient condition for $u_{t}$
to be analytic. There are two tensions when seeking such a suitable set $%
\mathcal{F}_{0}$: First, we would like to choose $\mathcal{F}_{0}$ as large
as possible in order to guarantee convergence of $\hat{u}_{t}$ over a large
set of functions. But at the same time we need to restrict $\mathcal{F}_{0}$
so that it satisfies $B\left( \mathcal{F}_{0}\right) \subseteq \mathcal{F}%
_{0}$. Second, to ensure a strong convergence result, we would like to
choose the norm $\left\Vert \cdot \right\Vert _{\mathcal{F}}$ as "strong" as
possible, e.g., as the $\sup $ norm. But establishing $T_{0}>0$ then proves
more difficult.

One way of designing the function class $\mathcal{F}_{0}$ is to build it
from the so--called eigenfunctions of $B$. Eigenfunctions are defined in
terms of the so--called spectrum of $B$,%
\begin{equation*}
\sigma \left( B\right) =\left\{ \xi \in \mathbb{C}:\left( \xi I-B\right) 
\text{ is not a bijection}\right\} \subset \left\{ \xi \in \mathbb{C}:Re\left( \xi \right) <0\right\} \cup \left\{ 0\right\} .
\end{equation*}%
In particular, for any given eigenvalue $\xi \in \sigma \left( B\right) $
there exists a corresponding eigenfunction $\phi \in \mathcal{D}\left(
B\right) $ so that $\left( \xi I-B\right) \phi =0\Leftrightarrow B\phi =\xi
\phi $. This in turn implies that $\phi \in \mathcal{D}\left( B^{\infty
}\right) $ with $B^{m}\phi =\xi ^{m}\phi $. Thus, 
\begin{equation*}
E_{t}\phi \left( x\right) =e^{\xi t}\phi \left( x\right) =\sum_{m=0}^{\infty
}\frac{t^{m}}{m!}B^{m}\phi \left( x\right) ,
\end{equation*}%
which is clearly analytic and so our power series expansion will converge
for any eigenfunction. The following corollary shows that in principle $%
\mathcal{F}_{0}$ can be chosen as the span of any given countable set of
eigenfunctions:

\begin{corollary}
\label{Cor: eigenfct}For any given sequence $\left\{ \left( \xi _{i},\phi
_{i}\right) \right\} _{i=1}^{\infty }$ of eigenpairs of $B$,%
\begin{equation*}
\mathcal{F}_{0}=\left\{ f=\sum_{i=1}^{\infty }\alpha _{i}\phi
_{i}:\sum_{i=1}^{\infty }\left\vert \alpha _{i}\right\vert <\infty \right\}
\subseteq \mathcal{D}\left( B^{\infty }\right) .
\end{equation*}
\end{corollary}

This particular choice of $\mathcal{F}_{0}$ is in some cases somewhat
restrictive in the sense that it may be only a small subset of $\mathcal{D}%
\left( B^{\infty }\right) $. However, in the special case of a given
semi--group's spectrum being countable, we generally have that $\mathcal{F}%
_{0}=\mathcal{D}\left( B^{\infty }\right) $. One example of this is
so--called time reversible Markov processes whose spectra are countable with
the corresponding eigenfunctions forming an orthnormal basis of $\mathcal{F}$%
; see, e.g., \cite{Hansen1998}. But many Markov processes are irreversible
and have an uncountable spectrum in which case $\mathcal{F}_{0}$ is a proper
subset of $\mathcal{D}\left( B^{\infty }\right) $.

The corollary does not guarantee that for any $f\in \mathcal{F}_{0}$ the
corresponding $u_{t}\left( x\right) $ is analytic -- only that it is
infinitely differentiable. To see the complications of ensuring analyticity,
observe that, for any given $f\in \mathcal{F}_{0}$ with $\mathcal{F}_{0}$
defined above, $B^{m}f=\sum_{i=1}^{\infty }\alpha _{i}\xi _{i}^{m}\phi _{i}$%
, $m\geq 1$, so that%
\begin{equation*}
\hat{u}_{t}=\sum_{m=0}^{M}\frac{t^{m}}{m!}B^{m}f=\sum_{i=1}^{\infty }\left(
\sum_{m=0}^{M}\frac{\left( \xi _{i}t\right) ^{m}}{m!}\right) \alpha _{i}\phi
_{i}.
\end{equation*}%
Thus,%
\begin{equation*}
\left\Vert u_{t}-\hat{u}_{t}\right\Vert _{\mathcal{F}}\leq \sup_{i\geq
1}\left\vert \sum_{m=0}^{M}\frac{\left( \xi _{i}t\right) ^{m}}{m!}-e^{-\xi
_{i}t}\right\vert \sum_{i=1}^{\infty }\left\vert \alpha _{i}\right\vert
\left\Vert \phi _{i}\right\Vert _{\mathcal{F}},
\end{equation*}%
and so we need at a minimum $\sup_{i\geq 1}\left\vert \sum_{m=0}^{M}\frac{%
\left( \xi _{i}t\right) ^{m}}{m!}-e^{-\xi _{i}t}\right\vert \rightarrow 0$, $%
M\rightarrow \infty $. But this convergence result will generally not hold;
for example, if $\xi _{i}\in \mathbb{R}$ and $\xi _{i}\rightarrow \infty $
as $i\rightarrow \infty $ then convergence will fail.

In conclusion, to ensure convergence, we need to impose restrictions on the
eigenvalues/the spectrum. We will now present such a set of conditions.
These will involve the so--called resolvent of the generator defined as%
\begin{equation*}
R\left( \xi \right) =\left( \xi I-B\right) ^{-1}\text{, }\xi \notin \sigma
\left( B\right) .
\end{equation*}

\begin{theorem}
\label{Th: analytic E}$t\mapsto E_{t}f\left( x\right) $ is analytic for all $%
t>0$ and all functions $f\in \mathcal{F}$ if and only if the following two
conditions are satisfied: There exists $0<\delta <\pi /2$ and $C_{B}<\infty $
so that%
\begin{equation}
\sigma \left( B\right) \subseteq \bar{\sigma}_{\delta }:=\left\{ \xi \in 
\mathbb{C}:\left\vert \arg \left( \xi \right) \right\vert >\pi /2+\delta
\right\} ,  \label{eq: sigma(A) cond}
\end{equation}%
and, for all $\varepsilon >0$, 
\begin{equation}
\left\Vert R\left( \xi \right) \right\Vert _{\mathrm{op}}\leq \frac{C_{B}}{%
\left\vert \xi \right\vert }\text{ for }\xi \in \mathbb{C}\backslash \bar{%
\sigma}_{\delta }.  \label{eq: R(z) cond}
\end{equation}

In particular, if $f\in E_{\tau _{0}}\left( \mathcal{F}\right) $ for some $%
\tau _{0}>0$ (that is, $f\left( x\right) =E_{\tau _{0}}g\left( x\right) $
for some $g\in \mathcal{F}$) then%
\begin{equation*}
\forall t<\frac{\tau _{0}}{C_{B}e}:\left\Vert u_{t}-\hat{u}_{t}\right\Vert _{%
\mathcal{F}}\rightarrow 0,\text{ \ \ }M\rightarrow \infty .
\end{equation*}
\end{theorem}

The first part of the theorem states necessary and sufficient conditions for 
$E_{t}f\left( x\right) $ to be analytic at any given $t>0$ and for \textit{%
any} $f\in \mathcal{F}$. The conditions (\ref{eq: sigma(A) cond})--(\ref{eq:
R(z) cond}) ensure that the spectrum of $B$ is such that the convergence
problem discussed before the theorem does not occur. This is a strong result
but at the same time (\ref{eq: sigma(A) cond})--(\ref{eq: R(z) cond}) are
rather strong conditions. Moreover, they tend to be difficult to verify in
practice since this requires knowledge of the spectrum $\sigma \left(
B\right) $. Primitive sufficient conditions for them to hold are provided in
the next section. Both the conditions and the results are relative to the
chosen function space and norm $\left( \mathcal{F},\left\Vert \cdot
\right\Vert _{\mathcal{F}}\right) $. By choosing $\mathcal{F}$ suitably
small, we expect that (\ref{eq: sigma(A) cond})--(\ref{eq: R(z) cond}) will
hold in great generality. We give an example of this in Section \ref{Sec:
conv bounded}.

The second part then shows that for the subclass of functions $f$ that
satisfy $f\left( x\right) =E_{\tau _{0}}g\left( x\right) $, for some $\tau
_{0}>0$ and $g$, analyticity of $E_{t}f\left( x\right) $ extends to $t=0$.
This part follows as a direct consequence of the first part since this
implies that $E_{t}f\left( x\right) =E_{t+\tau _{0}}g\left( x\right) $ is
analytic at $t=0$. The lower bound of the radius of convergence $T_{0}$
depends on the degree of smoothness of $f$, as measured by $\tau _{0}$, and
the properties of the model, specifically the bound $C_{A}$ on its resolvent.

The requirement $f\in E_{\tau _{0}}\left( \mathcal{F}\right) $ is difficult
to verify in a given application. In the leading case of $E_{t}f\left(
x\right) =\mathbb{E}\left[ f\left( x_{t}\right) |x_{0}=x\right] $, the
condition amounts to showing that there exists a solution $g\left( x\right) $
to the following integral equation $f\left( x\right) =\int g\left( y\right)
p_{\tau _{0}}\left( y|x\right) dy$ for some $\tau _{0}>0$, assuming that $%
x_{t}$ has a transition density $p_{t}\left( y|x\right) $. This is a
so--called Fredholm equation of the first kind; conditions for a solution to
this to exist are available but not easily verified in a given application.
However, it can be shown that, for any given $\tau _{0}$, $E_{\tau
_{0}}\left( \mathcal{F}\right) $ is dense in $\mathcal{F}$, see, e.g.,
Theorem 7.4.4 in \cite{davies2007}, and so the result will hold for "almost
every" $f\in \mathcal{F}$.

\subsection{Application to Jump-diffusions\label{sec: conv jump-diffusion}}

We now apply the general theory to our jump--diffusion model. In the
following, let $x_{t}$ be a weak solution to (\ref{eq: model}) for a given
specification of $\left( \mu ,\sigma ^{2},\lambda ,\nu \right) $ with
generator $A$ given in (\ref{eq: A def})--(\ref{eq: A_J def}) and $%
E_{t}f\left( x\right) =\mathbb{E}\left[ f\left( x_{t}\right) |x_{0}=x\right] 
$.

We first need to get a handle on the generator of the process and its domain 
$\mathcal{D}\left( A\right) $. A complete characterization of $\mathcal{D}%
\left( A\right) $ is unfortunately not possible and we will instead only
work with a subset of $\mathcal{D}\left( A\right) $ where the generator
takes the form (\ref{eq: A def}). Let $\mathcal{C}^{m}=\mathcal{C}^{m}\left( 
\mathcal{X}\right) $ denote the space of functions $f\left( x\right) $ with
domain $\mathcal{X}$ that are $m\geq 0$ times continuously differentiable
w.r.t. $x$. If $f\in \mathcal{C}^{2}$ then Ito's Lemma for jump--diffusions
(see, e.g., \cite{Cont2003}, Proposition 8.14) yields 
\begin{eqnarray*}
f\left( x_{t}\right) &=&f\left( x_{0}\right) +\int_{0}^{t}A_{D}f\left(
x_{s}\right) ds+\sum_{i:0\leq \tau _{i}\leq t}\left[ f\left( x_{\tau
_{i}^{-}}+\Delta x_{i}\right) -f\left( x_{\tau _{i}^{-}}\right) \right] \\
&&+\sum_{i=1}^{d}\int_{0}^{t}\frac{\partial f\left( x_{s}\right) }{\partial
x_{i}}\sigma _{i}\left( x_{s}\right) dW_{s},
\end{eqnarray*}%
where $A_{D}$ is defined in (\ref{eq: A_D def}), $\sigma _{i}\left( x\right)
=\left[ \sigma _{i1}\left( x\right) ,...,\sigma _{id}\left( x\right) \right] 
$ while $\tau _{i}$ and and $\Delta x_{i}$ denote the time and the size,
respectively, of the $i$th jump. Assuming $E_{t}\left\vert f\right\vert
\left( x\right) <\infty $ and $E_{t}(\frac{\partial f}{\partial x_{i}}\sigma
_{i})^{2}\left( x\right) <\infty $, $i=1,...,d$, we can take conditional
expectations w.r.t. the natural filtration on both sides of the above to
obtain (\ref{eq: A extend def}) with $A$ given in (\ref{eq: A def}). Thus,
the following is a subset of the domain of the generator,%
\begin{equation*}
\mathcal{D}_{0}\left( A\right) :=\left\{ f\in \mathcal{C}^{2}:E_{t}\left%
\vert f\right\vert \text{ and }E_{t}\left\Vert \frac{\partial f}{\partial x}%
\sigma \right\Vert ^{2}\text{ exist for all }t>0\right\} \subseteq \mathcal{D%
}\left( A\right) .
\end{equation*}

In the following we will only consider functions situated in $\mathcal{D}%
_{0}\left( A\right) $ and so not distinguish between the general generator
and the one restricted to $\mathcal{D}_{0}\left( A\right) $. Under the
assumption that $\mu $, $\sigma ^{2}$ and $\lambda $ and $f$ all belong to $%
\mathcal{C}^{2m}$, we can apply Ito's Lemma repeatedly and it follows
straightforwardly that%
\begin{equation*}
\mathcal{D}_{0}\left( A^{m}\right) :=\left\{ f\in \mathcal{C}%
^{2m}:E_{t}\left\vert A^{k}f\right\vert \text{ and }E_{t}\left\Vert \frac{%
\partial A^{k}f}{\partial x}\sigma \right\Vert ^{2}\text{ exist for all}t>0,%
\text{ }0\leq k\leq m-1\right\} \subseteq \mathcal{D}\left( A^{m}\right) .
\end{equation*}%
Implicit in this definition is the requirement that $\int_{\mathbb{R}%
^{d}}\left\vert A^{k}f\left( x+c\right) \right\vert \nu _{t}\left( c\right)
dc<\infty $ for $k=0,...,m$. Thus, a given $f\in \mathcal{C}^{2m}$ belongs
to $\mathcal{D}_{0}\left( A^{m}\right) $ if relevant moments w.r.t the jump
measure $\nu $ and the probability measure of $x_{t}$ exist. For example, if 
$f$ and all its derivatives are bounded, $\mu $, $\sigma ^{2}$ and $\lambda $
and all their derivatives are bounded by some function $V\left( x\right)
\geq 0$ with $\mathbb{E}\left[ V\left( x_{t}\right) \right] <\infty $, and $%
\nu $ has bounded support then $f\in \mathcal{D}_{0}\left( A^{\infty
}\right) $. Similarly, if $f$ is a polynomial of order $q$, $\mu $, $\sigma
^{2}$ and $\lambda $ are linear w.r.t. $x$, $\nu _{t}$ has all polynomial
moments, and $\mathbb{E}\left[ \left\Vert x_{t}\right\Vert ^{q}\right]
<\infty $, $0\leq t\leq T$, then $f\in \mathcal{D}\left( A^{\infty }\right) $%
.

Since $f\in \mathcal{D}_{0}\left( A^{\infty }\right) $ is necessary for our
expansion to work, we will maintain the following assumption on the model:

\begin{description}
\item[A.1] (i) $\mu $, $\sigma ^{2}$ and $\lambda $ belong to $\mathcal{C}%
^{\infty }$ and (ii) $\sup_{x}\lambda \left( x\right) <\infty $.
\end{description}

Part (i) ensures that, under suitable moment conditions as described above,
if $f\in \mathcal{C}^{\infty }$ then $f\in \mathcal{D}_{0}\left( A^{\infty
}\right) $. Part (ii) is imposed to simplify subsequent arguments since it
entails the following result (see Pazy, 1983, Theorem 3.2.1):

\begin{lemma}
\label{Lem: A_J bounded}Suppose that $A_{J}$ is a bounded operator. If $%
A_{D} $ generates an analytic semi--group then $A_{D}+A_{J}$ also generates
an analytic semi--group. Under A.1(ii), $A_{J}$ is a bounded operator.
\end{lemma}

Thus, for a given jump--diffusion model satisfying A.1(ii), or any other
conditions ensuring $A_{J}$ is bounded, we only need to ensure that the
diffusive component is analytic. In the following, we will implicitly assume
that indeed $A_{J}$ is bounded and derive conditions under which $E_{t}$ for
pure diffusion processes ($A_{J}=0$) is analytic.

Ideally we would now provide primitive conditions for general
jump--diffusion processes to satisfy the high-level conditions found in the
theorems and corollaries stated in the previous section. This is
unfortunately not possible since the spectral properties of jump--diffusions
are still not fully understood. We will therefore only state results for
special cases for which results do exist. At the same time we would like to
emphasise that we expect the results to hold more broadly.

We first develop conditions under which polynomial moment functions are
analytic. We start out with a few definitions: For a given multi-index $%
\alpha =\left( \alpha _{1},...\alpha _{d}\right) \in \mathbb{N}_{0}^{d}$ and 
$x=\left( x_{1},...,x_{d}\right) ^{\prime }\in \mathbb{R}^{d}$ let $%
\left\vert \alpha \right\vert =\alpha _{1}+\cdots +\alpha _{d}$ and $\left(
x\right) ^{\alpha }=x_{1}^{\alpha _{1}}\cdots x_{d}^{\alpha _{d}}$. We then
let%
\begin{equation*}
\mathcal{P}_{k}=\left\{ p\left( x\right) =\sum_{\left\vert \alpha
\right\vert \leq k}c_{\alpha }\left( x\right) ^{\alpha }:\alpha \in \mathbb{N%
}_{0}^{d},c_{\alpha }\in \mathbb{R}\right\}
\end{equation*}%
denote the family of polynomials of order $k$ and $\mathcal{P}_{k|\mathcal{X}%
}$ be these polynomials restricted to the domain of $x_{t}$. Observe here
that $\mathcal{P}_{k|\mathcal{X}}$ is a finite-dimensional function space.
In particular, we can choose a set of basis functions $e=\left(
e_{1},...,e_{N}\right) \in \mathcal{P}_{k|\mathcal{X}}$, where $N=\dim 
\mathcal{P}_{k|\mathcal{X}}$, so for any $p\in \mathcal{P}_{k|\mathcal{X}}$
there exists $c=\left( c_{1},...c_{N}\right) $ so that%
\begin{equation*}
p\left( x\right) =\sum_{i=1}^{N}c_{i}e\left( x\right) =c^{\prime }e\left(
x\right) .
\end{equation*}%
If $\mathcal{P}_{k|\mathcal{X}}$ satisfies the two conditions of Theorem \ref%
{Th: analytic} then analyticity follows automatically from the fact that
when we restrict the domain of $A$ to $\mathcal{P}_{k|\mathcal{X}}$ then it
becomes a finite--dimensional operator and therefore bounded:

\begin{corollary}
\label{Cor: poly analytic}Suppose that $x_{t}$ is a polynomial process in
the sense that, for all $k\geq 1$, $\mathcal{P}_{k|\mathcal{X}}\subseteq 
\mathcal{D}\left( A\right) $ and $A\left( \mathcal{P}_{k|\mathcal{X}}\right)
\subseteq \mathcal{P}_{k|\mathcal{X}}$. Then, for any $k\geq 1$ and any $%
p=c^{\prime }e\in \mathcal{P}_{k|\mathcal{X}}$, $u_{t}\left( x\right)
=E_{t}p\left( x\right) $ is analytic with radius $+\infty $ and satisfies
for all $x\in \mathcal{X}$,%
\begin{equation*}
u_{t}\left( x\right) =c^{\prime }\exp \left( t\bar{A}\right) e\left(
x\right) =c^{\prime }\sum_{m=0}^{\infty }\frac{t^{m}}{m!}\bar{A}^{m}e\left(
x\right) ,
\end{equation*}%
where $\bar{A}=\left[ \bar{a}_{ij}\right] _{1\leq i,j\leq N}\in \mathbb{R}%
^{N\times N}$ is defined as the solution to 
\begin{equation*}
Ae_{i}=\sum_{i=1}^{N}\bar{a}_{ij}e_{j},\text{ \ \ }i=1,...,N.
\end{equation*}

A sufficient condition for $x_{t}$ to be polynomial process is that $\mu \in 
\mathcal{P}_{1|\mathcal{X}}$, $\sigma ^{2}\in \mathcal{P}_{2|\mathcal{X}}$
and $\lambda \in \mathcal{P}_{2|\mathcal{X}}$.
\end{corollary}

The second result provides primitive conditions for the high--level
assumptions (\ref{eq: sigma(A) cond})--(\ref{eq: R(z) cond}) to hold in the
context of jump diffusions, where $E_{t}^{\ast }$ and $A^{\ast }$ denotes
the so--called adjoint operators of $E_{t}$ and $A$, respectively:

\begin{corollary}
\label{Cor: reversible}Suppose that $E_{t}f\left( x\right) =\mathbb{E}\left[
f\left( x_{t}\right) |x_{0}=x\right] $ and that $\mathcal{F}$ is a Hilbert
space with inner product $\left\langle \cdot ,\cdot \right\rangle $ so that $%
\left\Vert f\right\Vert _{\mathcal{F}}^{2}=\left\langle f,f\right\rangle $.
Suppose furthermore that $x_{t}$ is a time--reversible Markov process in the
sense that its generator is self--adjoint, $A=A^{\ast }$ (or, equivalently, $%
E_{t}=E_{t}^{\ast }$). Then (\ref{eq: sigma(A) cond})--(\ref{eq: R(z) cond})
are satisfied and so $t\mapsto E_{t}f\left( x\right) $ is analytic for all $%
t>0$.

Suppose that $x_{t}$ is a stationary diffusion process which satisfies the
conditions given in either Example 1, 2 or 3 in \cite{Scheinkman1995}. Then $%
x_{t}$ is time--reversible.
\end{corollary}

The time--reversibility condition implies that $A$'s spectrum is discrete
and contained in the negative half--line which suffices for (\ref{eq:
sigma(A) cond})--(\ref{eq: R(z) cond}) to hold. The three examples referred
to in the second part of the last theorem are time--homogenous scalar
diffusions, multivariate factor diffusion models, and a restricted class of
multivariate diffusions; see \cite{Scheinkman1995} for the precise details.

Note here that the corollary imposes no smoothness conditions on $\mu $, $%
\sigma ^{2}$ and $f$. This is because that $Af$ may still be well--defined
even without smoothness, c.f. above discussion of $\mathcal{D}\left(
A\right) $. However, its particular form in these cases is generally unknown
to us. Thus, in order to compute $Af$ in practice we restrict ourselves to
smooth models, as in Assumption A.1(i), and smooth choices of $f$, as in $%
\mathcal{D}_{0}\left( A\right) $.

Our third result again uses Theorem \ref{Th: analytic} but focuses on a
different class of "test functions" to obtain results for such models. We
restrict the function set to%
\begin{equation}
\mathcal{F}_{0}=\left\{ f\in \mathcal{C}^{\infty }:\sup_{\left\vert \alpha
\right\vert \geq 0}\left\Vert \partial _{x}^{\alpha }f\right\Vert _{\mathcal{%
F}}<\infty \right\} ,  \label{eq: F_0 def}
\end{equation}%
where $\partial _{x}^{\alpha }f=\partial ^{\alpha }f/\left( \partial
x^{\alpha }\right) $, which we equip with the norm $\left\Vert f\right\Vert
_{\mathcal{F}_{0}}=\sup_{\left\vert \alpha \right\vert \geq 0}\left\Vert
\partial _{x}^{\alpha }f\right\Vert _{\mathcal{F}}.$Importantly, if $f\in 
\mathcal{F}_{0}$ then, for any $\alpha \in \mathbb{N}_{0}^{\infty }$, $%
\partial _{x}^{\alpha }f\in \mathcal{F}_{0}$ with $\left\Vert \partial
_{x}^{\alpha }f\right\Vert _{\mathcal{F}_{0}}\leq \left\Vert f\right\Vert _{%
\mathcal{F}_{0}}.$This property of $\mathcal{F}_{0}$ ensures that if $\mu $
and $\sigma $ in $\mathcal{F}_{0}$ then $Af\in \mathcal{F}_{0}$ for all $%
f\in \mathcal{F}_{0}$ and so $\mathcal{F}_{0}\subseteq \mathcal{D}\left(
A^{\infty }\right) $. Moreover, the generator, when restricted to $\mathcal{F%
}_{0}$, is bounded and so the radius of convergence is infinite:

\begin{theorem}
\label{Th: analytic T 1}Suppose that $\mu $ and $\sigma ^{2}$ lie in $\left( 
\mathcal{F}_{0},\left\Vert \cdot \right\Vert _{\mathcal{F}_{0}}\right) $
defined in (\ref{eq: F_0 def}). Then $\mathcal{F}_{0}\subseteq \mathcal{D}%
\left( A\right) $ and $A:\mathcal{F}_{0}\mapsto \mathcal{F}_{0}$ is a
bounded operator. In particular, there exists $\bar{A}<\infty $ so that for
any $f\in \mathcal{F}_{0}$ and any $t>0$,%
\begin{equation*}
\left\Vert u_{t}-\hat{u}_{t}\right\Vert _{\mathcal{F}_{0}}\leq \frac{\left( t%
\bar{A}\right) ^{M+1}}{\left( M+1\right) !}\left\Vert f\right\Vert _{%
\mathcal{F}_{0}}\rightarrow 0\text{ as }M\rightarrow \infty .
\end{equation*}
\end{theorem}

Note here that convergence holds for all $t>0$ and that the convergence rate
is super-geometric. Moreover, the result allows for a broad class of
non-linear multivariate diffusion models. On the other hand, it rules out
unbounded drift and diffusion terms.

\subsection{Convergence over bounded sets\label{Sec: conv bounded}}

The above results are strong in the sense that they guarantee convergence
w.r.t a function norm over the full state space $\mathcal{X}$. But at the
same time they are restrictive in that they do not apply to general
multivariate jump--diffusion models. One way of allowing for a broader class
of models and functions is to restrict attention to solutions defined on a
bounded subset of $\mathcal{X}$ leading to the following class of so--called
localized Cauchy problems. We here focus on the case of pure diffusions
since for this class of models results exist on analytic solutions on
bounded sets.

Let $\mathcal{X}_{0}\subseteq \mathcal{X}$ be a bounded open set and let $%
u_{t}^{\ast }\left( x\right) $ be a function chosen by the researcher which
satisfies $u_{0}^{\ast }\left( x\right) =f\left( x\right) $. We then
consider the following "trimmed" version of the Cauchy problem for diffusion
models:%
\begin{eqnarray}
\partial _{t}w_{t}\left( x\right)  &=&A_{D}w_{t}\left( x\right) \text{ for }%
\left( t,x\right) \in \left( 0,\infty \right) \times \mathcal{X}_{0},
\label{eq: Dirichlet 1} \\
w_{t}\left( x\right)  &=&u_{t}^{\ast }\left( x\right) \text{ for}\left(
t,x\right) \in \left( 0,\infty \right) \times \mathcal{X}\backslash \mathcal{%
X}_{0},  \label{eq: Dirichlet 2}
\end{eqnarray}%
with initial condition $w_{0}\left( x\right) =f\left( x\right) $ for $x\in 
\mathcal{X}_{0}$. We now only require the solution $w_{t}\left( x\right) $
to solve the Cauchy problem on a bounded open subset $\mathcal{X}_{0}$ of
the full domain $\mathcal{X}$ and then pin down its behaviour outside of $%
\mathcal{X}_{0}$ through the pre-specified function $u^{\ast }$. The class
of problems on the form (\ref{eq: Dirichlet 1})--(\ref{eq: Dirichlet 2}) can
be described by a semi--group $E_{t}$ so that $w_{t}=E_{t}f$. By choosing $%
\mathcal{X}_{0}$ as a bounded set, the requirements for the semi--group to
be analytic becomes a lot less restrictive and essentially requires $\mu $
and $\sigma ^{2}$ to be sufficiently smooth; see, e.g., Chapter 3 in \cite%
{Lunardi1995}. The following theorem states the precise conditions:

\begin{theorem}
\label{Th: analytic bounded}Suppose that $\mu \left( x\right) $, $\sigma
^{2}\left( x\right) $ and $u_{t}^{\ast }\left( x\right) $ are analytic
functions so that, for some $0<C_{0},C_{1}<\infty $,%
\begin{equation}
\left\Vert \partial _{x}^{\alpha }\mu \left( x\right) \right\Vert \leq
C_{0}C_{1}^{-1-\left\vert \alpha \right\vert }\left\vert \alpha \right\vert
!,\text{ \ \ \ }\left\Vert \partial _{x}^{\alpha }\sigma ^{2}\left( x\right)
\right\Vert \leq C_{0}C_{1}^{-1-\left\vert \alpha \right\vert }\left\vert
\alpha \right\vert !,\text{ \ \ }x\in \mathcal{X}_{0};
\label{eq: mu sig bounds}
\end{equation}%
and, for some $c>0$ and for all $x,y\in \mathcal{X}_{0}$, $y^{\prime }\sigma
^{2}\left( x\right) y\geq c\left\Vert y\right\Vert $. Then, $w_{t}:\mathcal{X%
}_{0}\mapsto \mathbb{R}$ is analytic at any $t>0$ w.r.t. the uniform norm, $%
\left\Vert w_{t}\right\Vert _{\mathcal{F},0}=\sup_{x\in \mathcal{X}%
_{0}}\left\vert w_{t}\left( x\right) \right\vert $.

Suppose furthermore that $f\left( x\right) =E_{\tau _{0}}g\left( x\right) $
for some continuous function $g$. Then, $w_{t}:\mathcal{X}_{0}\mapsto 
\mathbb{R}$ is analytic at $t=0$ with radius of convergence $T_{0}>1/\left(
\rho \tau _{0}\right) $, where $\rho =\rho \left( B,d\right) \in (0,1]$.
\end{theorem}

This provides simple and relatively weak conditions under which a series
expansion of $w_{t}$ will converge. But will such series expansion be a good
approximation to $u_{t}$? By eq. (\ref{eq: Dirichlet 1}) together with the
initial condition 
\begin{equation*}
\left. \partial _{t}^{m}w_{t}\left( x\right) \right\vert
_{t=0^{+}}=A_{D}^{m}f\left( x\right) ,\text{ \ \ }x\in \mathcal{X}_{0}\text{%
, \ \ }m\geq 0.
\end{equation*}%
Thus, under the conditions of the theorem, our proposed power series
approximation shares derivatives with $u_{t}$ on $\mathcal{X}_{0}$. At the
same time, the solution $w_{t}$ will generally differ from the global
solution $u_{t}$. However, if we restrict $f\in \mathcal{D}^{\infty }\left(
A_{D}\right) $ then $\left. \partial _{t}^{m}u_{t}\left( x\right)
\right\vert _{t=0^{+}}=A_{D}^{m}f\left( x\right) $ and so $w_{t}\left(
x\right) =u_{t}\left( x\right) $, $x\in \mathcal{X}_{0}$, and the power
series will be consistent on $\mathcal{X}_{0}$. In particular, if we can
show that $w_{t}\left( x\right) $ is analytic on $\mathcal{X}_{0}$ then the
same will hold for $u_{t}\left( x\right) $ when considered as a function
with domain $\mathcal{X}_{0}$. This result combined with Lemma \ref{Lem: A_J
bounded} shows that our power series expansions converges for a very broad
class of diffusion models over bounded subsets of their domains.

\subsection{Expansion of "irregular" moments}

Finally, we provide an analysis of smoothed expansions on the form (\ref{eq:
E approx general}). First, by following the same arguments as in Theorem \ref%
{Th: w-hat error bound}, it is easily shown using Taylor's Theorem that if $%
u_{0,s}$ satisfies A.0 then $\hat{u}_{t}\left( x\right) =\hat{E}_{t}f\left(
x\right) $ given in (\ref{eq: E approx general}) with $M_{1}+M_{2}=M$
satisfies%
\begin{equation*}
\left\vert u_{t}\left( x\right) -\hat{u}_{t}\left( x\right) \right\vert \leq
\sum_{m_{1}+m_{2}=M}\frac{\left( -s\right) ^{m_{1}}t^{m_{2}}}{m_{1}!m_{2}}%
\sup_{0\leq s,\tau \leq t}\left\vert \left( \partial _{s}\right)
^{m_{1}}B^{m_{2}}E_{\tau }u_{0,s}\left( x\right) \right\vert =O\left(
s^{M}\right) +O\left( t^{M}\right) .
\end{equation*}%
One could now hope for that as long as $u_{0,s}\left( x\right) $ is
sufficiently regular then the expansion would converge under conditions
similar to the ones in the "regular" case analyzed in the previous section.
This is unfortunately not the case. To see this, observe that in order for
the expansion to be asymptotically valid $s\mapsto u_{0,s}\left( x\right) $
has to be analytic so that%
\begin{eqnarray*}
f\left( x\right)  &=&\sum\limits_{m_{1}=0}^{\infty }\frac{\left( -s\right)
^{m_{1}}}{m_{1}!}\partial _{s}^{m}u_{0,s}\left( x\right)  \\
&=&\sum\limits_{m_{1}=0}^{M_{1}}\frac{\left( -s\right) ^{m_{1}}}{m_{1}!}%
\partial _{s}^{m}u_{0,s}\left( x\right) +\sum\limits_{m_{1}=M_{1}+1}^{\infty
}\frac{\left( -s\right) ^{m_{1}}}{m_{1}!}\partial _{s}^{m}u_{0,s}\left(
x\right)  \\
&=&:\hat{f}\left( x\right) +r\left( x\right) ,
\end{eqnarray*}%
where $r\left( x\right) $ is the remainder term from a $M_{1}$th order
Taylor expansion of $s\mapsto u_{0,s}$ around $s=0$. If $f\notin \mathcal{D}%
\left( B\right) $ and, for some $M_{1}\geq 1$, $\hat{f}\left( x\right) \in 
\mathcal{D}\left( B\right) $ then obviously $r\left( x\right) \notin 
\mathcal{D}\left( B\right) $. Thus, as $M$ grows large enough, we must have $%
\hat{f}\notin \mathcal{D}\left( B\right) $ in which case $\hat{E}_{t}f\left(
x\right) =\sum_{m_{2}=0}^{M_{2}}\frac{t^{m_{2}}}{m_{2}}B^{m_{2}}\hat{f}%
\left( x\right) $ is not well-defined. In practice, we expect $\hat{E}%
_{t}f\left( x\right) $ in (\ref{eq: E approx general}) to become numerically
unstable as $M_{2}\rightarrow \infty $. That is, the numerical error will
start blowing up.

This demonstrates that the proposed series expansions of irregular functions
such as densities and option prices should be used with caution: As more
terms are added to the expansions, they will most eventually become
numerically unstable and produce unreliable estimates. However, as we shall
see in the next section, the expansions still work well when a reasonably
small number of terms are used.

\section{Numerical results\label{sec: numeric}}

We assess the performance of our approximations when applied to the problem
of option pricing when the underlying asset's dynamics are described by a
stochastic volatility model with jumps under the risk--neutral measure. We
consider the following class of asset pricing models where the log--price $%
s_{t}$ of a given asset exhibits both stochastic volatility and jumps, 
\begin{equation}
ds_{t}=\left( \mu -v_{t}/2-\lambda \left( v_{t}\right) \bar{J}\right) dt+%
\sqrt{v_{t}}dW_{1t}+\log \left( J_{t}+1\right) dN_{t},  \label{eq: logret}
\end{equation}%
where the volatility process $v_{t}$ is solution to either 
\begin{equation}
dv_{t}=\kappa _{V}\left( \alpha _{V}-v_{t}\right) dt+\sigma _{V}v_{t}^{\beta
}\left( \rho dW_{1t}+\sqrt{1-\rho ^{2}}dW_{2t}\right) ,  \label{eq: vol}
\end{equation}%
or 
\begin{equation}
d\log v_{t}=\kappa _{V}\left( \alpha _{V}-\log v_{t}\right) dt+\sigma
_{V}\left( \rho dW_{1t}+\sqrt{1-\rho ^{2}}dW_{2t}\right) .
\label{eq: logvol}
\end{equation}%
Here, $\mu =r-\delta $ where $r$ and $\delta $ are the risk-free rate and
the constant dividend, respectively. To ensure that the model has a
well-defined solution, $\kappa _{V}$, $\alpha _{V}$, $\sigma _{V}$ are
restricted to be positive and $1/2\leq \beta \leq 1$.

The jump component consists of a Cox process $N_{t}$ with a jump intensity
function given by $\lambda \left( v\right) =\lambda _{0}+\lambda _{1}v$, and
a random variable $J_{t}$ with support $[-1,\infty )$, and expectation $\bar{%
J}$. We include $-\lambda \left( v_{t}\right) \bar{J}$ in the drift as a
compensator such that the jump part is a martingale. For example, if $J+1$
is chosen to be log-normally distributed with parameters $\mu _{J}$ and $%
\sigma _{J}$, then $\bar{J}=\exp \left( \mu _{J}+\sigma _{J}^{2}/2\right) -1$%
. Special cases of this model include \citet{Merton1976}, where both
volatility and jump intensity are constant, $v_{t}=\sigma _{0}$ and $\lambda
\left( v\right) =\lambda _{0}$. Eq. (\ref{eq: logret}) together with either (%
\ref{eq: vol}) or (\ref{eq: logvol}) is a special case of (\ref{eq: model})
with $x_{t}=\left( s_{t},v_{t}\right) $.

This class of models subsumes models in \cite{andersen2002} and \cite%
{Wan2021} as well as a number of other special cases. Compared to \cite%
{andersen2002}, our specification allows the variance process to be the
non-affine continuous-time GARCH model ($\beta =1$) and the CEV model ($%
1/2<\beta <1$). Also, compared to \cite{Wan2021}, we allow for
state-dependent jump intensity ($\lambda _{1}\neq 0$) which they rule out.

We consider a European call option with payoff $f\left( s_{T}\right) \equiv
\max \left\{ \exp \left( s_{T}\right) -K,0\right\} $ at maturity time $T>0$,
where $K=100$ is the strike price. With the above model formulated under the
so--called risk--neutral measure, let $u_{\Delta }\left( s,v\right) =E\left[
f\left( s_{T}\right) |s_{T-\Delta }=s,v_{T-\Delta }=v\right] $ be the
expected risk--neutral pay-off the option expires in $\Delta $ time units
and the current log stock price and volatility is $s$ and $v$, respectively.
Within the above class of models for $s_{t}$, no closed-form formula for the
option price is available. We here implement our proposed series expansion
of the unknown price, $\hat{u}_{t}\left( s,v\right) $, as given in (\ref{eq:
u-hat special}), where we choose $u_{0,\Delta }$ as the pay-off under the
Black- Scholes model as given in (\ref{eq: B-S pay-off}).

In the case of state--dependent jumps, we need to compute the integration
part of $A_{J}$ using numerical methods. Since $\log \left( J_{t}+1\right) $
is i.i.d. with normal distribution with mean $m_{J}$ and standard deviation $%
\sigma _{J}$ for all models in this section, we use the Gauss-Hermite
quadrature with different numbers of nodes and weights, whose values are
fixed after choosing the number of nodes and weights, c.f. Section \ref{sec:
numeric}.

To assess the numerical performance of our expansion, we will use as
benchmark the option price obtained via Monte Carlo methods, where the total
number of simulation trials $10,000,000$ and the time-step is $10,000$ per
year, see Chapter 3 in \cite{Giesecke2018} for details. We measure the
accuracy of the approximations by the maximum absolute error and the
absolute percentage error defined as follows: $\max_{S\in \left[ 90,110%
\right] }\lvert \hat{u}_{\Delta }\left( s,v\right) -u_{\Delta }^{MC}\left(
s,v\right) \rvert $ and $\max_{S\in \left[ 90,110\right] }\lvert \hat{u}%
_{\Delta }\left( s,v\right) -u_{\Delta }^{MC}\left( s,v\right) \rvert
/u_{\Delta }^{MC}\left( s,v\right) $, respectively, where $\hat{u}_{\Delta
}\left( s,v\right) $ and $u_{\Delta }^{MC}\left( s,v\right) $ are the series
expansion and the Monte Carlo version of the option price, respectively.

We consider increasingly challenging experiments, aiming to assess the
resilience of our method to the approximation of option prices under
increasingly complex models.

\subsection{State-independent jumps\label{subsec: indepjump}}

In this subsection, we explore the performance our method when jumps are
state-independent ($\lambda _{1}=0$).

In Figure \ref{fig: SQR_NJ}, we depict the approximation errors resulting
from our method for (\ref{eq: logret})--(\ref{eq: vol}) with $\beta =0.5$
across different levels of the current asset price. As in \cite{Wan2021},
the parameter values used in this experiment are chosen as the estimates
reported in \cite{Eraker2004}, which are displayed in the figure legend.
From left to right, the time to maturity ranges from $\Delta :=T-t=1/52$,
1/12, and 1/4, respectively. In the top three panels, the maximum absolute
error has been plotted for 1st, 2nd, 3rd, and 4th order approximation,
respectively; whereas the horizontal axis denotes the number of nodes and
weights for the Gauss-Hermite quadrature. For the bottom three panels, the
vertical axis denotes the absolute percentage error; whereas the horizontal
axis denotes the stock price.

We make the following observations: First, for all maturities, as $M$
increases, the approximation error decreases. Second, for a given order of
approximation, our method is more accurate as the time to maturity
decreases. Third, one can achieve accurate approximations with small number
of nodes and weights used in the quadrature approximation of the jump
component. For small time to maturity ($t=\Delta =1/52$), it is sufficient
to use the Gauss-Hermite quadrature with 10 nodes and weights, but, for
larger time to maturity ($t=\Delta =1/12$, 1/4), only 4 or 5 nodes and
weights. It indicates that the error in computing the integration part of
the jump component is smaller than the error of our Taylor series
approximation as the number of nodes and weights for the Gauss-Hermite
quadrature increases.

\begin{figure}
	\centering
	\includegraphics[width=0.7\linewidth]{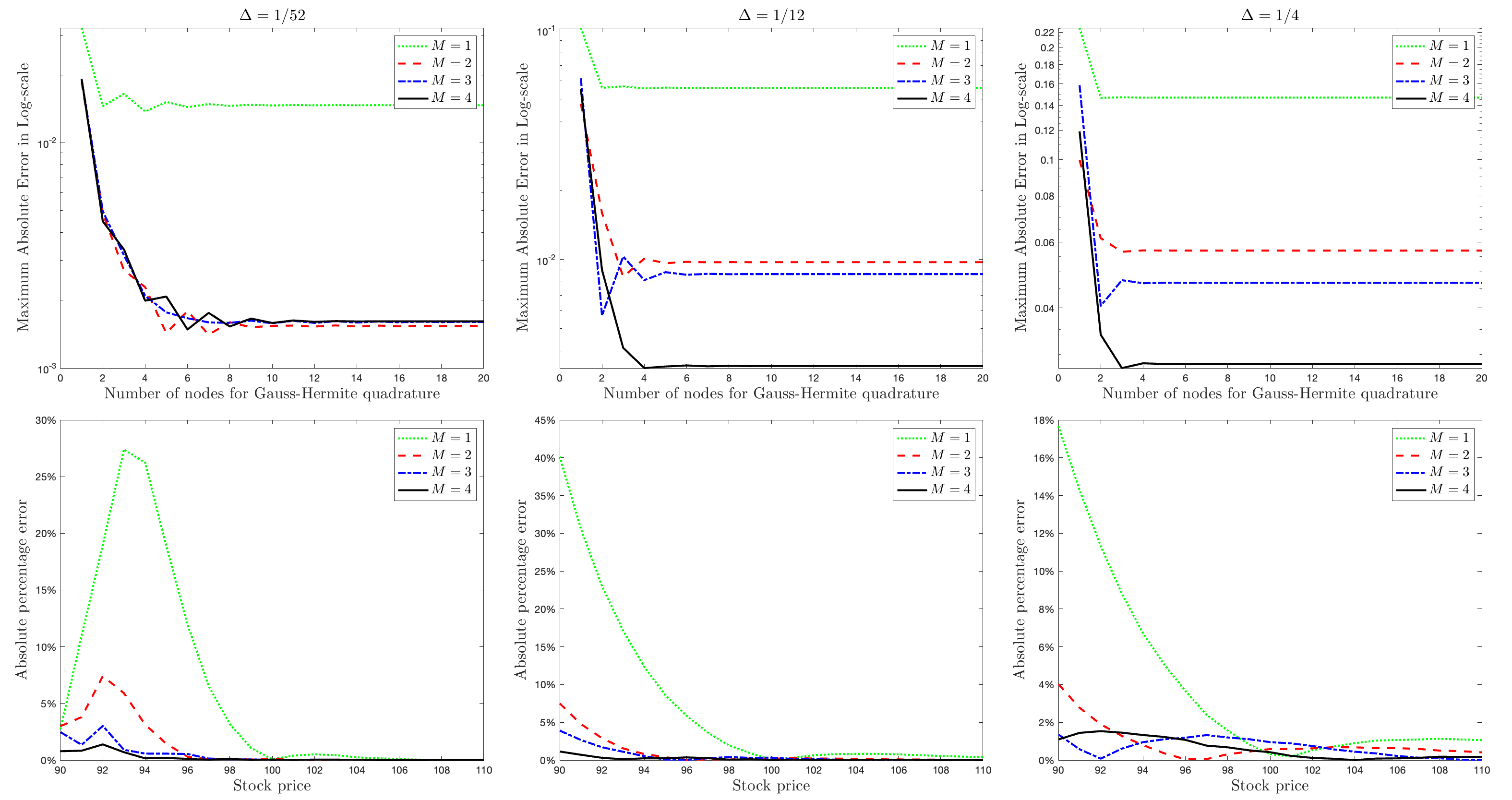}
	\caption[]{Maximum absolute error (top) and absolute percentage error (bottom) of option price approximations in (\ref{eq: logret})--(\ref{eq: vol}) with $\beta =0.5$, $M=1,2,3,4$, $\Delta=1/52,1/12,1/4$, $s=100$, $v=0.0416$}
\label{fig: SQR_NJ}
\end{figure}

Figure \ref{fig: CEV} and \ref{fig: GARCHNJ} investigate the numerical
performance of our approximation for the call option under the stochastic
volatility model (\ref{eq: logret})--(\ref{eq: vol}) with the CEV ($\beta
=0.8$) and GARCH ($\beta =1$) specifications of variance, respectively. The
parameters for the CEV and GARCH specification are from \cite{aitsahalia2007}
and \cite{Yang2017}, respectively, but we added or changed the parameters
for the jump part, which is the same as in \cite{Wan2021}.

For $\Delta =1/52$ and $\Delta =1/12$, the performances of the approximation
for both two models share three patterns arose in the outcome in Figure \ref%
{fig: SQR_NJ}. However, for longer time-to-maturity, $\Delta =1/4$, the
higher order of approximation does not guarantee smaller approximation
error. In general, the performance of the approximation error is good with
shorter maturities and/or $\beta $ takes on a relatively small value.

\begin{figure}
	\centering
	\includegraphics[width=0.7\linewidth]{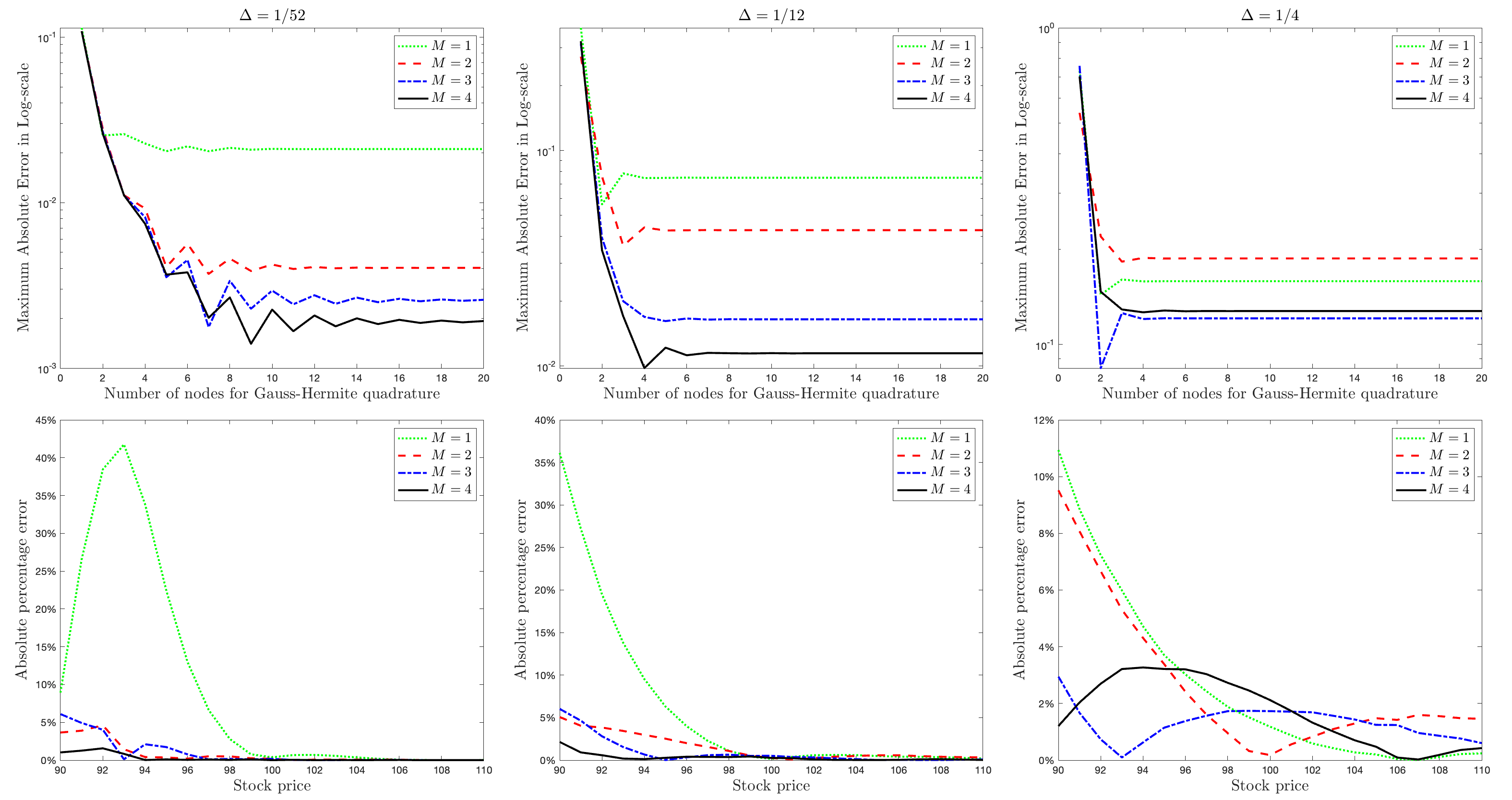}
	\caption[]{Maximum absolute error (top) and absolute percentage error (bottom) of option price approximations in (\ref{eq: logret})--(\ref{eq: vol}) with $\beta=0.8$}
\label{fig: CEV}
\end{figure}

\begin{figure}
	\centering
	\includegraphics[width=0.7\linewidth]{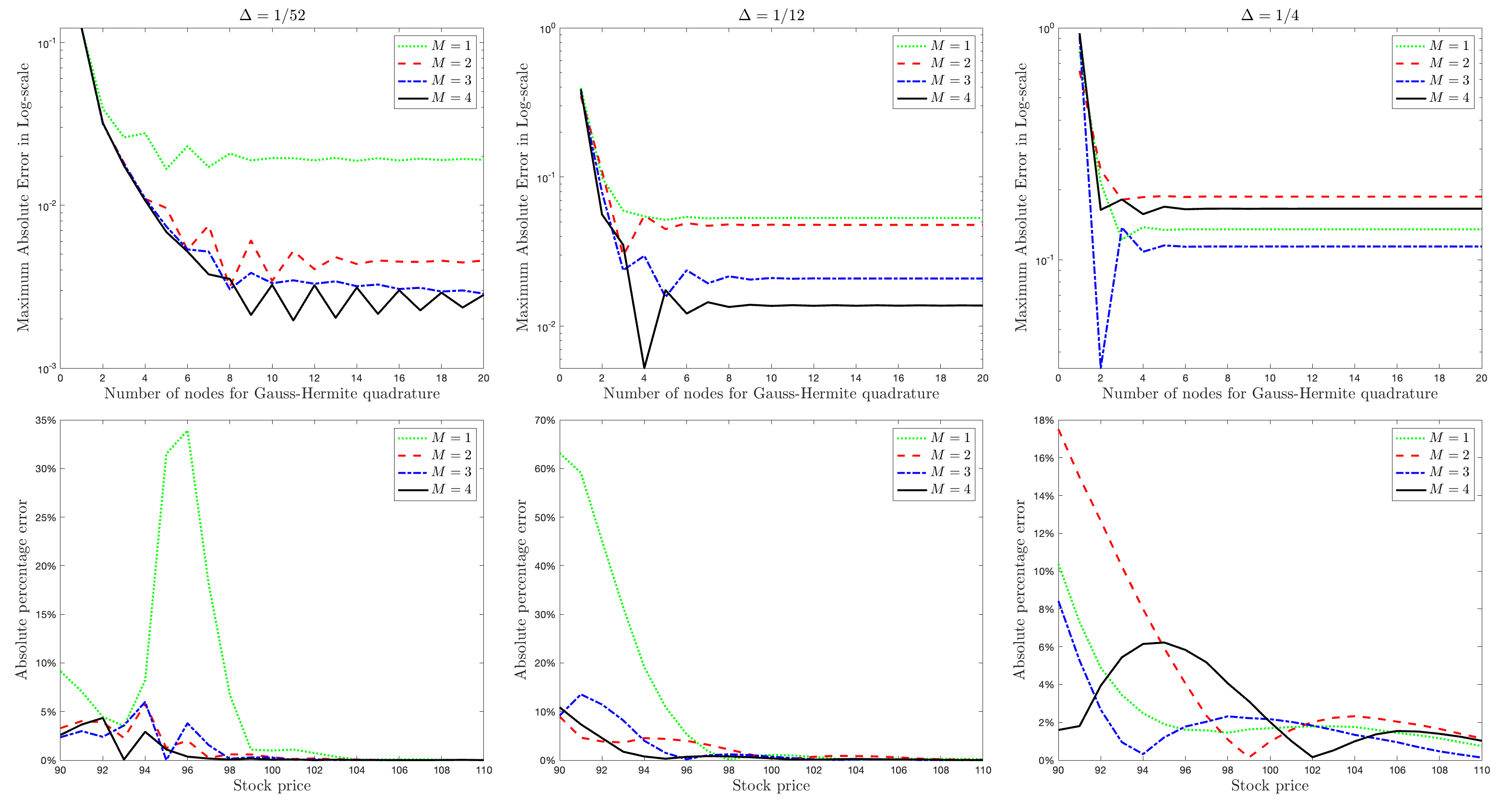}
	\caption[]{Maximum absolute error (top) and absolute percentage error (bottom) of option price approximations in (\ref{eq: logret})--(\ref{eq: vol}) with $\beta=1.0$}
\label{fig: GARCHNJ}
\end{figure}

\subsection{State-dependent jumps\label{subsec: depjump}}

In this subsection, we provide results for the case where the option price
is computed under models with state-dependent jump intensities ($\lambda
_{1}\neq 0$).

In Figures \ref{Fig: Lambda 1}--\ref{fig: Lambda 3}, we depict the relative
error of the approximation for the same three models considered in the
previous subsection, except that now $\lambda _{1}=1,10,30$, when
time-to-maturity equal to one month, $\Delta =1/12$. In each figure, from
left to right, the state dependency of jump intensity ranges $\lambda
_{1}=1,10,30$. Overall, the approximation errors for each of the three
models are comparable to that of the same model with state-independent jumps
($\lambda _{1}=0$). Furthermore, the absolute percentage error is smaller
for all orders of approximation for larger $\lambda _{1}$. It indicates that
the magnitude of $\lambda _{1}$ affects the level of option prices but does
not affect the approximation errors. That is, the performance of our
approximation is not very sensitive to the degree of state dependence of the
jumps as measured by the value of $\lambda _{1}$.

\begin{figure}
	\centering
	\includegraphics[width=0.7\linewidth]{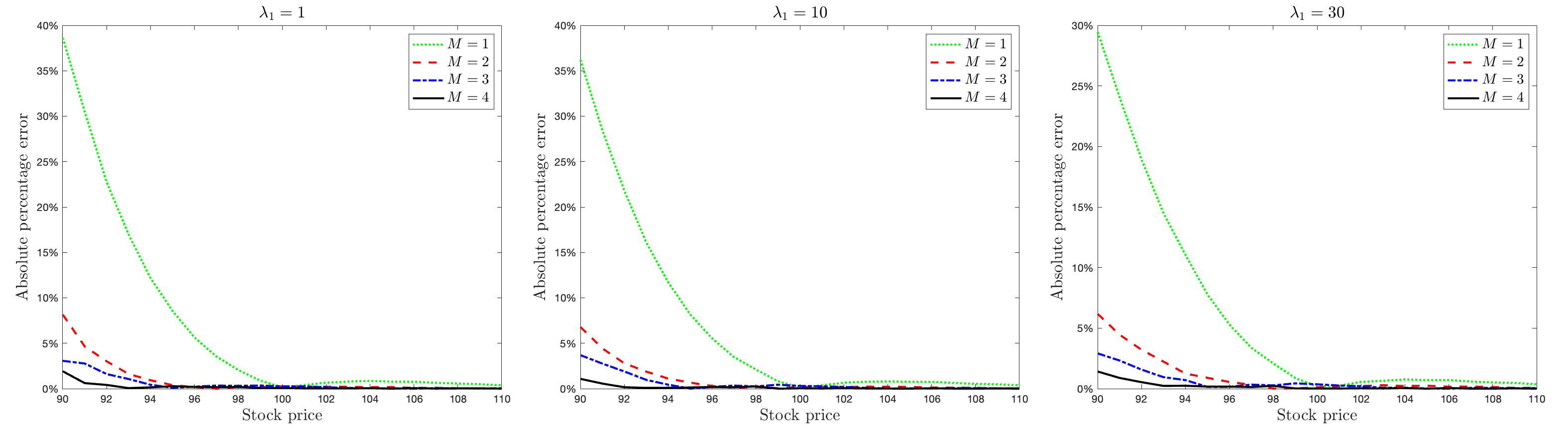}
	\caption[]{Maximum absolute error (top) and absolute percentage error (bottom) of option price approximations in (\ref{eq: logret})--(\ref{eq: vol}) with $\beta =0.5$ and $\lambda _{1}=1$ (left panel), $10$ (middle panel) and $30$ (right panel)}
\label{Fig: Lambda 1}
\end{figure}

\begin{figure}
	\centering
	\includegraphics[width=0.7\linewidth]{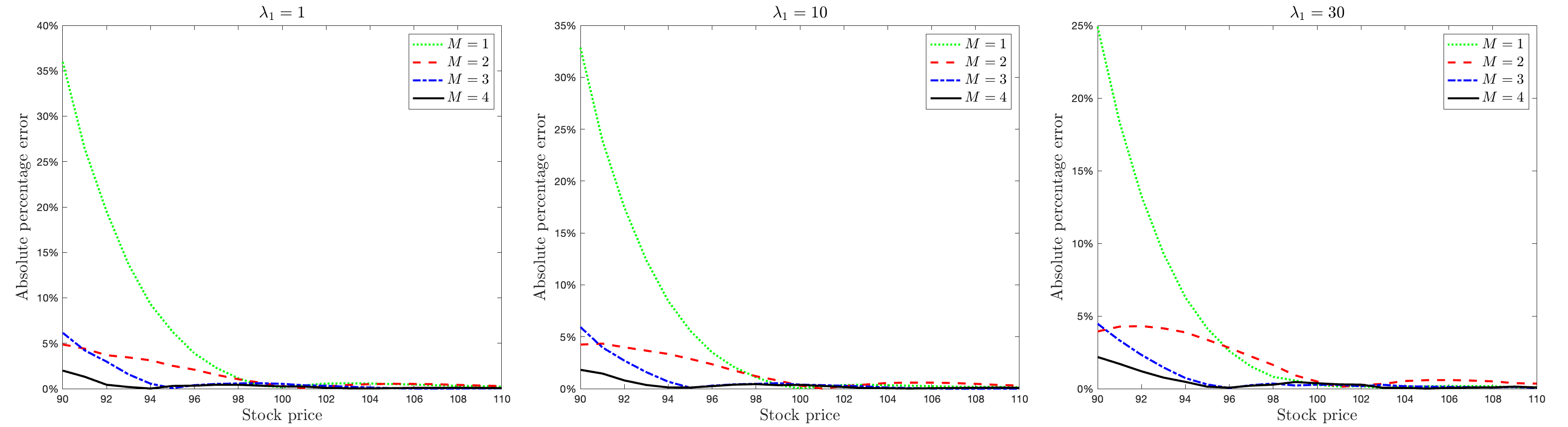}
	\caption[]{Maximum absolute error (top) and absolute percentage error (bottom) of option price approximations in (\ref{eq: logret})--(\ref{eq: vol}) with $\beta =0.8$ and $\lambda _{1}=1$ (left panel), $10$ (middle panel) and $30$ (right panel)}
\label{Fig: Lamda 2}
\end{figure}

\begin{figure}
	\centering
	\includegraphics[width=0.7\linewidth]{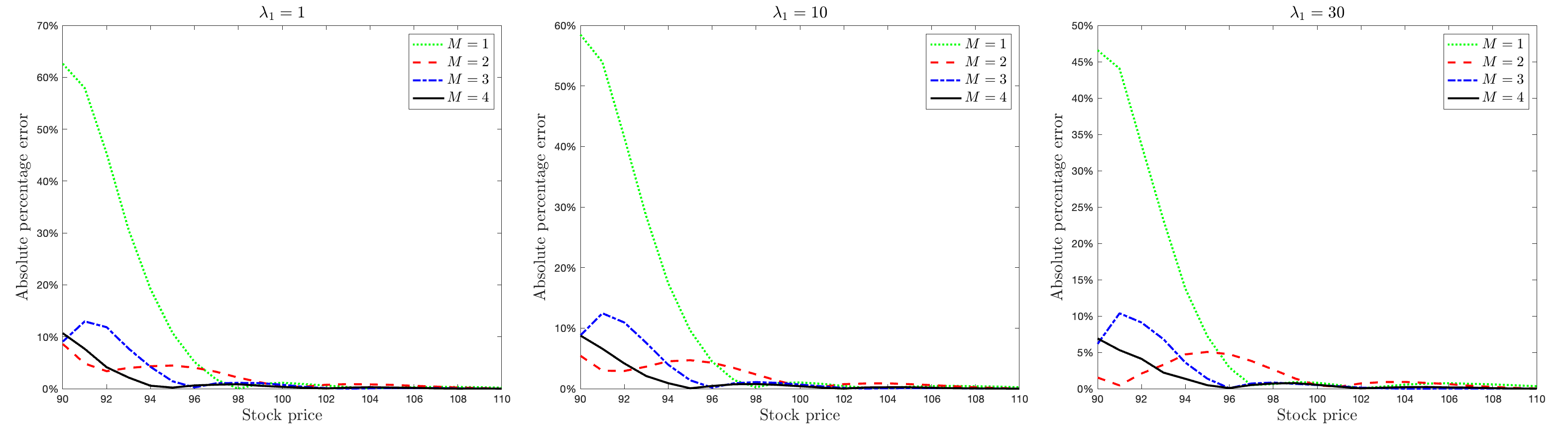}
	\caption[]{Maximum absolute error (top) and absolute percentage error (bottom) of option price approximations in (\ref{eq: logret})--(\ref{eq: vol}) with $\beta =1.0$ and $\lambda _{1}=1$ (left panel), $10$ (middle panel) and $30$ (right panel)}
\label{fig: Lambda 3}
\end{figure}

Next, we consider the performance when $v_{t}$ solves the log--volatility
model (\ref{eq: logvol}) with parameters chosen as $\left( r,\delta ,\kappa
_{V},\alpha _{V},\sigma _{V},\rho \right) =\left(
0.0304,0,0.0145,-0.8276,0.1153,-0.6125\right) $ and $\left( \lambda _{0},\mu
_{J},\sigma _{J}\right) =\left( 0.0137,-0.000125,0.015\right) $; these are
the estimates reported in \cite{andersen2002}. Figures \ref{fig: logV 1} and %
\ref{fig: logV 2} display the relative error of the approximation for the
call option under this model for different values of $\Delta $ and $\lambda
_{1}$ with $v_{0}=\alpha _{V}=-0.8276$. We see that even for the 2nd order
approximation, the approximation error is quite small for all choices of
time--to--maturity and $\lambda _{1}$. The plotted errors are now more
ragged which we conjecture is due to bigger numerical errors in the Monte
Carlo benchmark that we use for comparison.

\begin{figure}
	\centering
	\includegraphics[width=0.7\linewidth]{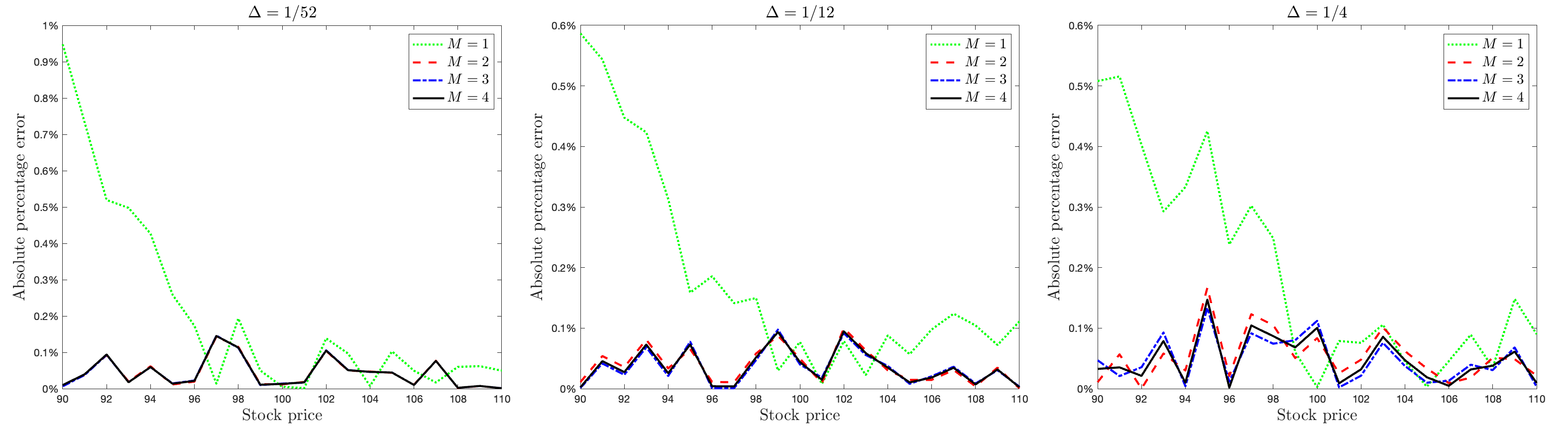}
	\caption[]{Maximum absolute error (top) and absolute percentage error (bottom) of option price approximations in (\ref{eq: logret}) and (\ref{eq: logvol}) with $\lambda _{1}=1$ and $\Delta=1/52$ (left), 1/12 (center), 1/4 (right)}
\label{fig: logV 1}
\end{figure}

\begin{figure}
	\centering
	\includegraphics[width=0.7\linewidth]{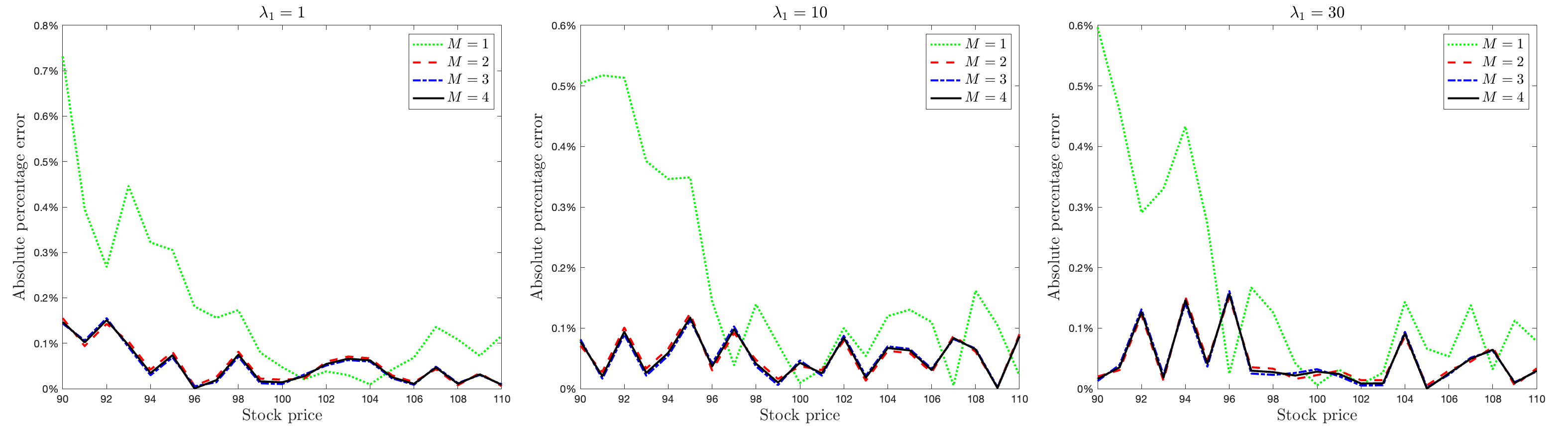}
	\caption[]{Maximum absolute error (top) and absolute percentage error (bottom) of option price approximations in (\ref{eq: logret}) and (\ref{eq: logvol}) with $\lambda _{1}=1$ (left), 10 (center), 30 (right) and $\Delta =1/12$}
\label{fig: logV 2}
\end{figure}

\subsection{Two-factor affine jump diffusion model \label{subsec: two}}

We here wish to examine the robustness of our method when applied to more
complex models that go beyond one-factor volatility. We consider the
stochastic volatility model with two factors for the volatility used in \cite%
{filipovic2016}. In their specification, the dynamics of $s_{t}$ under the
risk--neutral measure are given by 
\begin{equation}
\begin{split}
ds_{t}& =\left( \mu -v_{t}/2-\lambda \left( v_{t},m_{t}\right) \bar{J}%
\right) dt+\sqrt{v_{t}}dW_{1t}+\log \left( J_{t}^{S}+1\right) dN_{t}, \\
dv_{t}& =\kappa _{V}\left( m_{t}-v_{t}\right) dt+\sigma _{V}\sqrt{v_{t}}%
\left( \rho dW_{1t}+\sqrt{1-\rho ^{2}}dW_{2t}\right) +J_{t}^{V}dN_{t}, \\
dm_{t}& =\kappa _{m}\left( \alpha _{m}-m_{t}\right) dt+\sigma _{m}\sqrt{m_{t}%
}W_{3t},
\end{split}
\label{eq: twofactor}
\end{equation}%
where $W_{1}$, $W_{2}$, and $W_{3}$ are mutually independent standard
Brownian motions. Compared to the models of the previous subsection, there
is a second variance factor $m_{t}$, which represents a stochastic level
around which $v_{t}$ reverts. The jump component consists of: (i) $N_{t}$, a
Cox process with a bounded intensity function given by $\lambda \left(
v,m\right) =\lambda _{0}+\lambda _{1}v+\lambda _{2}m$, and the variance jump
size $J^{V}$ is exponentially distributed with parameter $\mu _{J}^{V}=%
\mathbb{E}\left[ J_{t}^{V}\right] $.

Figure \ref{fig: SVJFGM} reports the performance of our approximation for
different times to maturity, and with different numbers of nodes for
Gauss-Hermite quadrature. The parameter values we used are estimates in \cite%
{aitsahalia2020}. The performance of the approximation shares the same
patterns that we found in Figure \ref{fig: SQR_NJ}. It indicates that the
performance of the approximation is still very good when we add more factors
to the volatility specification.

\begin{figure}
	\centering
	\includegraphics[width=0.7\linewidth]{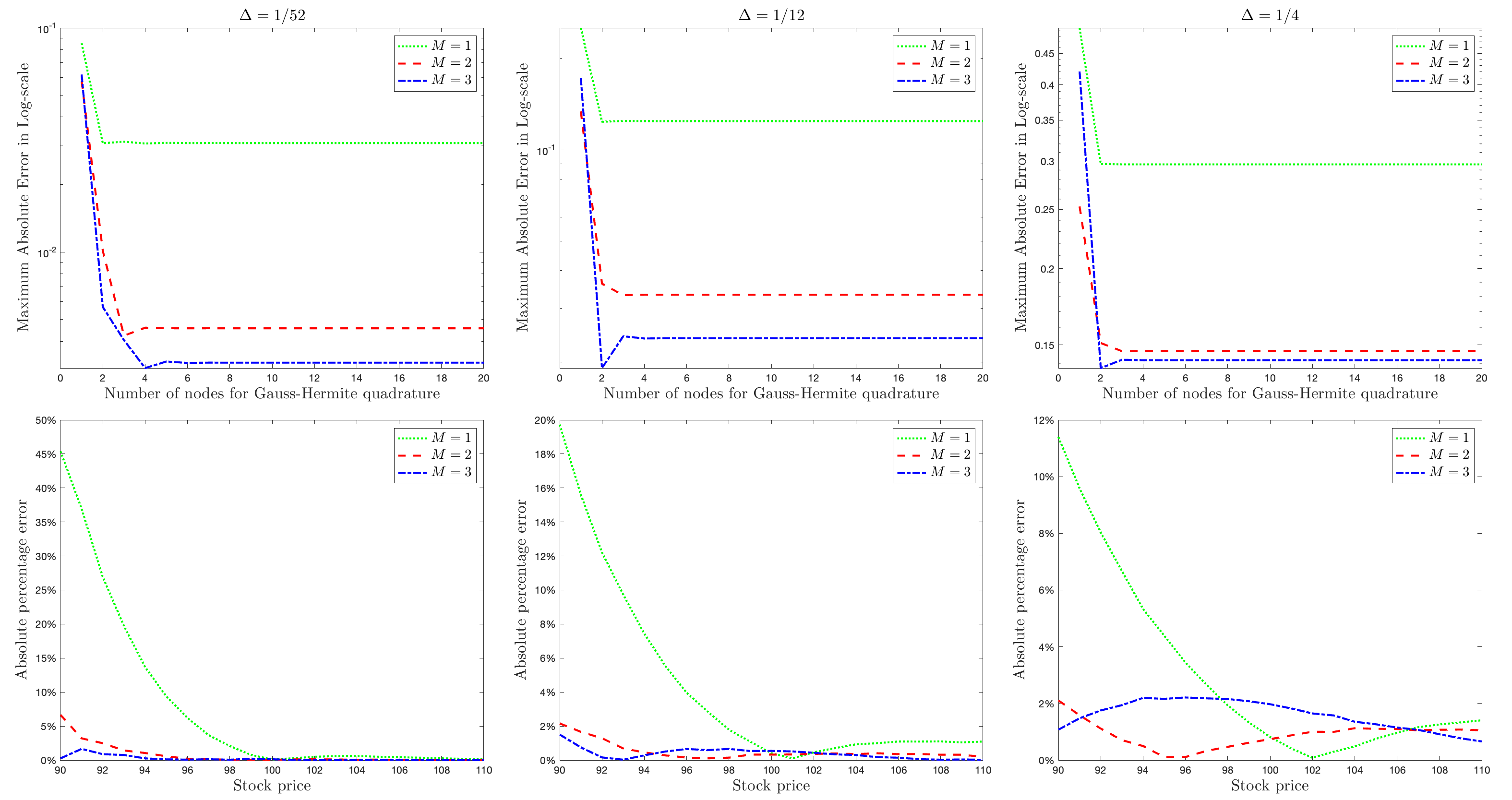}
	\caption[]{Maximum absolute error (top) and absolute percentage error (bottom) of option price approximations in (\ref{eq: twofactor}) with $\Delta =1/52$, 1/12, 1/4}
\label{fig: SVJFGM}
\end{figure}

\section{Conclusion\label{sec: conclusion}}

This paper provides a general framework for developing and analyzing series
expansions of moments of continuous-time Markov processses, including
jump-diffusions. The expansions come in two versions depending on the
features of the moment. For "regular" moments, we provide conditions under
which the corresponding expansion will converge towards the actual moments
as more terms are added. For the "smoothed" expansion, no such theoretical
guarantees exist: The expansion will eventually become imprecise as the
number of terms grows. A numerical study shows that the smoothed expansions
still work well in practice when a relatively small number of terms are used
in its implementation.

\newpage

\bibliographystyle{chicago}

\newpage

\appendix

\section{Relationship to existing literature\label{sec: relationship}}

We here first present the proposal of \cite{kristensen2011} and show that it
falls within the general framework of Section \ref{sec: atd}. We then
proceed to show that the class of series expansions of \cite{kristensen2011}
contains as special cases the ones of \cite{Yangetal2019} and \cite{Wan2021}

Recall the definition of $u_{t}\left( x\right) $ in (\ref{eq: u FK}) of the
motivating jump--diffusion example. To approximate $u_{t}\left( x\right) $, 
\cite{kristensen2011} takes as starting point an auxiliary model on the form%
\begin{equation}
dx_{0,t}=\mu _{0}\left( x_{0,t}\right) dt+\sigma _{0}\left( x_{0,t}\right)
dW_{t}+J_{0,t}dN_{0,t},  \label{eq: aux model}
\end{equation}%
where $N_{0,t}$ is a Poisson process with jump intensity $\lambda _{0}\left(
x\right) $ and $J_{0,t}$ has density $\nu _{0}\left( \cdot |x\right) $. Let $%
u_{0,t}\left( x\right) $ be the solution to the problem of interest but now
under the auxiliary model,%
\begin{equation}
-\partial _{t}u_{0,t}\left( x\right) =\left[ A_{0}-r\left( x\right) \right]
u_{0,t}\left( x\right) ,  \label{eq: PIDE aux}
\end{equation}%
with initial condition $u_{0,0}\left( x\right) =f\left( x\right) $, where $A$
has been replaced by the auxiliary model's generator, $A_{0}=A_{0,D}+A_{0,J}$%
, with%
\begin{eqnarray*}
A_{0,D}f\left( x\right)  &=&\sum_{i=1}^{d}\mu _{0,i}\left( x\right) \partial
_{x_{i}}f\left( x\right) +\frac{1}{2}\sum_{i,j=1}^{d}\sigma
_{0,ij}^{2}\left( x\right) \partial _{x_{i},x_{j}}^{2}f\left( x\right) , \\
A_{0,J}f\left( x\right)  &=&\lambda _{0}\left( x\right) \int_{\mathbb{R}^{d}}%
\left[ f\left( x+c\right) -f\left( x\right) \right] \nu _{0}\left(
c|x\right) dc.
\end{eqnarray*}%
\cite{kristensen2011} then subtract (\ref{eq: PIDE aux}) from (\ref{eq: PIDE
true}) and, after some straightforward manipulations, arrive at the
following PIDE of $\Delta u_{t}\left( x\right) \equiv u_{t}\left( x\right)
-u_{0,t}\left( x\right) $:%
\begin{equation}
\partial _{t}\Delta u_{t}\left( x\right) =\left[ A-r\left( x\right) \right]
\Delta u_{t}\left( x\right) +d_{t}\left( x\right) .  \label{PDE approx}
\end{equation}%
where%
\begin{equation}
d_{t}\equiv \left( A-A_{0}\right) u_{0,t}.  \label{eq: d def}
\end{equation}%
Since the initial conditions of (\ref{eq: PIDE terminal}) and (\ref{eq: PIDE
aux}) are the same, the initial condition of (\ref{PDE approx}) becomes $%
\Delta u_{T}\left( x\right) =0$ which is now smooth and bounded. As with $%
u_{t}\left( x\right) $, $\Delta u_{t}\left( x\right) $ can be represented as
a moment function using Feynman-Kac formula under weak regularity conditions,%
\begin{equation}
u_{t}\left( x\right) =u_{0,t}\left( x\right) +\int_{0}^{t}E_{s}d_{s}\left(
x\right) ds.  \label{eq: Delta w moment}
\end{equation}%
The second term on the right-hand side of Eq. (\ref{eq: Delta w moment})
delivers an exact expression of the difference between $u_{t}\left( x\right) 
$ and $u_{0,t}\left( x\right) $.

The next step utilizes the smoothness of $d_{t}$ to obtain a Taylor
expansion w.r.t. time of this second term. We first develop a power series
expansion of the integrand,%
\begin{equation}
w_{s}\left( x\right) \equiv E_{s}d_{s}\left( x\right) ,\text{ \ \ }s\geq 0,
\label{eq: w D def}
\end{equation}%
at $s=t$ taking the form%
\begin{equation*}
\hat{w}_{s}\left( x\right) \equiv \sum\limits_{m=0}^{M-1}\frac{\left(
s-t\right) ^{m}}{m!}\left. \partial _{s}^{m}w_{s}\left( x\right) \right\vert
_{s=t}=\sum\limits_{m=0}^{M-1}\frac{\left( s-t\right) ^{m}}{m!}\left(
\partial _{t}+A-r\right) ^{m}d_{t}\left( x\right) .
\end{equation*}%
for some $M\geq 1$. This assumes that $\left( \partial _{t}+A-r\right)
^{m}d_{t}\left( x\right) $ is well-defined, $1\leq m\leq M$. Combining these
last two equations, substituting the resulting expression into (\ref{eq:
Delta w moment}) and evaluating the integral $\int\nolimits_{0}^{t}\hat{w}%
_{s}\left( x\right) ds$, we obtain the approximation originally proposed in 
\cite{kristensen2011}, here extended to the general case of jump--diffusions:%
\begin{equation}
\hat{u}_{t}\left( x\right) \equiv u_{0,t}\left( x\right)
+\sum\limits_{m=0}^{M-1}\frac{t^{m+1}}{\left( m+1\right) !}\left( -\partial
_{t}+A-r\right) ^{m}d_{t}\left( x\right) ,  \label{eq: u_M def}
\end{equation}%
where $\left( -\partial _{t}+A-r\right) ^{0}d_{t}\left( x\right)
=d_{t}\left( x\right) $.

Finally, observe that an equivalent representation of $\hat{u}_{t}\left(
x\right) $ is%
\begin{equation}
\hat{u}_{t}\left( x\right) =\sum\limits_{m=0}^{M}\frac{t^{m}}{m!}\left(
-\partial _{t}+A-r\right) ^{m}u_{0,t}\left( x\right)  \label{eq: u_M def 2}
\end{equation}%
which follows from combining (\ref{eq: PIDE aux}) and (\ref{eq: d def}) to
obtain%
\begin{equation*}
d_{t}=Au_{0,t}-A_{0}u_{0,t}=Au_{0,t}+\left\{ -\partial _{t}u_{0,t}-r\left(
x\right) u_{0,t}\right\} =\left( -\partial _{t}+A-r\right) u_{0,t}.
\end{equation*}%
We recognize (\ref{eq: u_M def 2}) as a special case of the general proposal
in (\ref{eq: E approx special}).

Next, we demonstrate that the above class of series expansions include as
special cases the approximate transition densities and option prices
proposed in \cite{Yangetal2019} and \cite{Wan2021}. With $r=0$ and $f\left(
x\right) =\delta \left( y-x\right) $, (\ref{eq: u_M def 2}) becomes%
\begin{equation}
\hat{p}_{t}\left( y|x\right) =\sum\limits_{m=0}^{M}\frac{t^{m}}{m!}\left(
-\partial _{t}+A\right) ^{m}p_{0,t}\left( y|x\right) ,  \label{eq: p-hat KL}
\end{equation}%
where $p_{0,t}\left( y|x\right) $ is the transition density of the auxiliary
model. Now, let us first consider the transition density expansion developed
in \cite{Yangetal2019} for pure diffusions ($A_{J}=0$). Inspecting the
expansion presented in eq. (10) of their paper, we recognize it to be
identical to above when $p_{0,t}$ is chosen as in eq. (\ref{eq: p_0 BM})
with $\sigma _{0}=\sigma \left( x\right) $. Thus, \cite{Yangetal2019} is a
special case of \cite{kristensen2011}. This somehow went unnoticed by the
authors and we here clarify the connection between the two papers. Second,
consider the expansion of the transition density in \cite{Wan2021} in the
pure diffusion case. As explained by the authors themselves, the preferred
version of the expansion used in this paper is the same as the series
expansion of Yang et al. (2019) when $\mu _{0}=0$ in the auxiliary BM\
model. And so the pure diffusion version of \cite{Wan2021} is also a special
case of \cite{kristensen2011}. 

Next, we show that the expansion of option prices developed in \cite{Wan2021}
is again a special case of \cite{kristensen2011}. Setting $r\left( x\right)
=0$ and $f\left( x\right) =\left( \exp \left( x_{1}\right) -K\right) ^{+}$
and using as auxiliary model (\ref{eq: BM aux}), $\hat{u}_{t}\left( x\right) 
$ as given in (\ref{eq: u_M def 2}) delivers an expansion of the expected
pay-off of a European option where $u_{0,t}\left( x\right) $ is\ now the
pay-off function under the Black--Scholes model. To connect this option
price approximation with the corresponding proposal of \cite{Wan2021},
observe that $u_{0,t}\left( x\right) =\int f\left( y\right) p_{0,t}\left(
y|x\right) dy$, where $p_{0,t}\left( y|x\right) $ is given in (\ref{eq: p_0
BM}). Substituting this into (\ref{eq: u_M def 2}) and changing the order of
integration and differentiation yields%
\begin{eqnarray}
\hat{u}_{t}(x) &=&\sum\limits_{m=0}^{M}\frac{t^{m}}{m!}\left( -\partial
_{t}+A\right) ^{m}\int f\left( y\right) p_{0,t}\left( y|x\right) dy  \notag
\\
&=&\int f\left( y\right) \left\{ \sum\limits_{m=0}^{M}\frac{t^{m}}{m!}\left(
-\partial _{t}+A\right) ^{m}p_{0,t}\left( y|x\right) \right\} dy  \notag \\
&=&\int f\left( y\right) \hat{p}_{t}\left( y|x\right) dy,
\label{eq: KM approx}
\end{eqnarray}%
where $\hat{p}_{t}\left( y|x\right) $ is the density approximation we
arrived at in (\ref{eq: p-hat KL}). Thus, for simple moment functions, such
as the ones appearing in European option prices with constant interest
rates, the expansion of \cite{kristensen2011} is equivalent to first
developing the corresponding expansion for the transition density and then
using this to compute the relevant moment. However, in practice, it is
easier to directly employ (\ref{eq: u_M def 2}) with $u_{0,t}$ chosen as the
pay-off under the Black--Scholes model since this avoids having to compute
the integral $\int f\left( y\right) \hat{p}_{t}\left( y|x\right) dy$ after
developing the expansion of the transition density.

Let us consider \cite{Wan2021}'s proposal for option pricing approximation:
They take as starting point that the pay-off can be written as $%
u_{t}(x)=\int f\left( y\right) p_{t}\left( y|x\right) dy,$and then replace $%
p_{t}\left( y|x\right) $ by the approximation given in (\ref{eq: p-hat KL})
with auxiliary model chosen as Brownian motion with drift. As we just
demonstrated in (\ref{eq: KM approx}), this is identical to the
approximation developed in \cite{kristensen2011} when the auxiliary model is
chosen as the Black--Scholes model since the log--price in this case follows
a Brownian Motion with drift. Thus, the option price approximation of \cite%
{Wan2021} is again a special case of \cite{kristensen2011}.

\section{Extension to time-inhomogenous problems \label{Sec: Time-inhomo}}

We here present the extension of our method to handle time--inhomogenous
models and problems where no closed-form solution is available to (\ref{eq:
PIDE aux}). As motivating example, consider the following extended version
of the model in (\ref{eq: model}):%
\begin{equation}
dx_{t}=\mu _{t}\left( x_{t}\right) dt+\sigma _{t}\left( x_{t}\right)
dW_{t}+J_{t}\left( x_{t}\right) dN_{t},  \label{eq: model 2}
\end{equation}%
where now $\lambda _{t}\left( x\right) $, $\mu _{t}\left( x\right) $, $%
\sigma _{t}\left( x\right) $ and $\nu _{t}\left( x_{t}\right) $ are now
allowed to vary with $t$. This in turn implies that the corresponding
generator is also time--varying, $A_{t}f\left( x\right) =A_{D,t}f\left(
x\right) +A_{J,t}f\left( x\right) $, where%
\begin{eqnarray*}
A_{D,t}f\left( x\right)  &=&\sum_{i=1}^{d}\mu _{i,t}\left( x\right) \partial
_{x_{i}}f\left( x\right) +\frac{1}{2}\sum_{i,j=1}^{d}\sigma
_{ij,t}^{2}\left( x\right) \partial _{x_{i},x_{j}}^{2}f\left( x\right) , \\
A_{J,t}f\left( x\right)  &=&\lambda _{t}\left( x\right) \int_{\mathbb{R}^{d}}%
\left[ f\left( x+c\right) -f\left( x\right) \right] \nu _{t}\left( c\right)
dc.
\end{eqnarray*}%
We are interested in computing $u_{s,t}\left( x\right) $ defined as%
\begin{equation}
u_{s,t}\left( x\right) =E_{s,t}f\left( x\right) ,\text{ \ \ }0\leq s\leq t,
\end{equation}%
where%
\begin{equation}
\left( s,t,f\right) \mapsto E_{s,t}f\left( x\right) \equiv \mathbb{E}\left[
\left. \exp \left( -\int\nolimits_{s}^{t}r\left( x_{u}\right) du\right)
f\left( x_{t}\right) \right\vert x_{s}=x\right] .
\end{equation}%
Due to the time--inhomogeneity, the operator $E_{s,t}f\left( x\right) $ is
now indexed by two time variables, $s$ and $t$. At the same time, for any
fixed value of $s\geq 0$, $\left( t,f\right) \mapsto E_{s,t}f\left( x\right) 
$ remains a semi--group when $\mathcal{F}$ is chosen suitably. Most of the
ideas and results from Sections \ref{sec: atd}--\ref{sec: theory} therefore
carry over to the time--inhomogenous case with only minor differences.
Below, we present the series expansion and explain how the theory applies to
this.

We take as starting point a given $\left( s,t,f\right) \mapsto
E_{s,t}f\left( x\right) $ where, for any given $s\geq 0$, $\left( t,f\right)
\mapsto E_{s,t}f\left( x\right) $ is assumed to be semi--group on some
funtion space $\left( \mathcal{F},\left\Vert \cdot \right\Vert _{\mathcal{F}%
}\right) $. In the following, we keep $s\geq 0$ fixed. We denote by $%
\mathcal{D}\left( B_{s}\right) $ the set of functions $f\in \mathcal{F}$ for
which there exists $g_{s}\in \mathcal{F}$ such that, for each $t\geq 0$,%
\begin{equation}
E_{s,t}f\left( x\right) =f\left( x\right) +\int_{s}^{t}E_{s,u}g_{s}\left(
x\right) du,
\end{equation}%
and we write $B_{s}f\left( x\right) :=g_{s}\left( x\right) $ and call $B_{s}$
the (extended) generator of $E_{s,t}$. In the motivating example above, it
is easily shown by Ito's Lemma that $B_{t}=A_{t}-r$ on the space%
\begin{equation*}
\mathcal{D}_{0}\left( B_{s}\right) :=\left\{ f\in \mathcal{C}^{2}\cap 
\mathcal{F}:E_{s,t}\left\vert f\right\vert \text{ and }E_{s,t}\left\Vert 
\frac{\partial f}{\partial x}\sigma _{s}\right\Vert ^{2}\text{ exist for all 
}t>0\right\} \subseteq \mathcal{D}\left( B_{s}\right) .
\end{equation*}
For any regular function $f$, regular in the sense that $f\in \mathcal{D}%
\left( B_{s}^{M}\right) $, we have

\begin{equation}
\partial _{t}u_{s,t}\left( x\right) =B_{t}^{m}u_{s,t}\left( x\right) ,\text{
\ }t>0,  \label{eq: PIDE time-var}
\end{equation}%
c.f. \cite{Rueschendorf2016}, which corresponds to the so--called forward
equation. Thus, in this case the following is a valid series expansion of $%
u_{s,t}\left( x\right) $:%
\begin{equation}
\hat{u}_{s,t}\left( x\right) =\sum\limits_{m=0}^{M}\frac{\left( t-s\right)
^{m}}{m!}B_{s}^{m}f\left( x\right) .
\end{equation}

If $f$ is irregular, so that $f\notin \mathcal{D}\left( B_{s}\right) $, we
introduce a smoothed version of it, $u_{0,s,t}\left( x\right) $ which is
assumed to satisfy:

\begin{description}
\item[A.0'] (i) $\lim_{t\rightarrow s^{+}}u_{0,s,t}\left( x\right) =f\left(
x\right) $ and (ii) $u_{0,s,t}\in \mathcal{D}\left( \left( \partial
_{t}\right) ^{M_{1}}\right) \cap \mathcal{D}\left( B_{s}^{M_{2}}\right) $
for some $M_{1},M_{2}\geq 1$.
\end{description}

Following the same steps as in the time--homogenous case of Section \ref%
{sec: atd}, we obtain the following series expansion:%
\begin{equation}
\hat{u}_{s,t}\left( x\right) =\sum\limits_{m=0}^{M}\frac{\left( t-s\right)
^{m}}{m!}\left( B_{s}-\partial _{t}\right) ^{m}u_{0,s,t}\left( x\right) .
\label{eq: u-hat time-inhomo}
\end{equation}

\section{Proofs \label{sec: Proofs}}

\begin{proof}[Proof of Theorem \protect\ref{Th: w representation}]
Use eq. (\ref{eq: w representation}) together with $%
E_{s}E_{t}=E_{s+t}=E_{t}E_{s}$ to obtain%
\begin{equation*}
E_{t}u_{s}\left( x\right) =u_{s}\left( x\right) +\int_{0}^{t}E_{s+w}\left(
Bf\right) \left( x\right) dw=\int_{0}^{t}E_{w}\left( E_{s}\left( Bf\right)
\right) \left( x\right) dw.
\end{equation*}

The second part of the theorem is obtained by taking derivatives w.r.t. $s$
on both sides of (\ref{eq: w representation}) and using that the right-hand
side derivative equals $E_{s}\left( Bf\right) \left( x\right) =Bu_{s}\left(
x\right) $ if this function is continuous w.r.t $s$ from the right.
\end{proof}

\bigskip

\begin{proof}[Proof of Theorem \protect\ref{Th: w-hat error bound}]
We expand $t\mapsto E_{t}f\left( x\right) $ around $E_{0}f\left( x\right)
=f\left( x\right) $ recursively: First rewrite (\ref{eq: w representation})
as%
\begin{equation}
E_{t}f\left( x\right) =f\left( x\right) +\int_{0}^{t}E_{t_{1}}\left(
Bf\right) \left( x\right) dt_{1}.  \label{eq: Dynkin}
\end{equation}%
Since $Bf\in \mathcal{D}\left( B\right) $ by assumption, we can apply (\ref%
{eq: Dynkin}) again to $E_{t_{1}}\left( Bf\right) \left( x\right) $ yielding 
\begin{equation*}
E_{t_{1}}\left( Af\right) \left( x\right) =Bf\left( x\right)
+\int_{0}^{t_{1}}E_{t_{2}}\left( B^{2}f\right) \left( x\right) dt_{2}.
\end{equation*}%
Substitute the right-hand side of the last equation into (\ref{eq: Dynkin})
to obtain%
\begin{eqnarray*}
E_{t}f\left( x\right) &=&f\left( x\right) +\int_{0}^{t}\left\{ Bf\left(
x\right) +\int_{0}^{t_{1}}E_{t_{2}}\left( B^{2}f\right) \left( x\right)
dt_{2}\right\} dt_{1} \\
&=&f\left( x\right) +tBf\left( x\right)
+\int_{0}^{t}\int_{0}^{t_{1}}E_{t_{2}}\left( B^{2}f\right) \left( x\right)
dt_{2}dt_{1}.
\end{eqnarray*}%
Repeating this argument $M$ more times yields the claimed result.
\end{proof}

\bigskip

\begin{proof}[Proof of Theorem \protect\ref{Th: analytic}]
By definition, $\left\Vert B^{m}f\right\Vert _{\mathcal{F}}/m!\leq
1/T_{0}^{m}$. Thus,%
\begin{equation*}
\left\Vert u_{t}-\hat{u}_{t}\right\Vert _{\mathcal{F}}\leq
\sum_{m=M+1}^{\infty }\frac{t^{m}}{m!}\left\Vert B^{m}f\left( x\right)
\right\Vert _{\mathcal{F}}\leq \sum_{m=M+1}^{\infty }\left( \frac{t}{T_{0}}%
\right) ^{m}=\frac{\left( t/T_{0}\right) ^{M+1}}{1-t/T_{0}}\rightarrow 0.
\end{equation*}
\end{proof}

\bigskip

\begin{proof}[Proof of Theorem \protect\ref{Th: analytic E}]
The first part follows from Theorem 2.5.2 of \cite{Pazy1983}. To show the
second part, recall the definition of radius of convergence $T_{0}$ in (\ref%
{eq: T_0 def}). To bound the right hand side of (\ref{eq: T_0 def}), first
use that $f\left( x\right) =E_{\tau _{0}}g\left( x\right) $ and that $A$ and 
$E_{\tau _{0}}$ commute to obtain $\left\Vert B^{m}f\right\Vert _{\mathcal{F}%
}=\left\Vert \left( BE_{\tau _{0}/m}\right) ^{m}g\right\Vert _{\mathcal{F}%
}\leq \left\Vert BE_{\tau _{0}/m}\right\Vert _{\mathrm{op}}^{m}\left\Vert
g\right\Vert _{\mathcal{F}}$. Next, due to (\ref{eq: sigma(A) cond})--(\ref%
{eq: R(z) cond}), we can apply part (d) of Theorem 2.5.2 of \cite{Pazy1983}
yielding $\left\Vert BE_{\tau _{0}/m}\right\Vert _{\mathrm{op}}^{m}\leq
\left( C_{A}m/\tau _{0}\right) ^{m}$. In total, 
\begin{equation*}
\left\Vert B^{m}f\right\Vert _{\mathcal{F}}/m!\leq \left\{ \left( \frac{C_{A}%
}{\tau _{0}}\right) ^{m}m^{m}/m!\right\} ^{1/m}\left\Vert g\right\Vert _{%
\mathcal{F}}^{1/m}\leq \left( \frac{C_{A}e}{\tau _{0}}\right) \left\Vert
g\right\Vert _{\mathcal{F}}^{1/m},
\end{equation*}%
and we conclude that $T_{0}\geq \tau _{0}/\left( C_{A}e\right) $.
\end{proof}

\bigskip

\begin{proof}[Proof of Corollary \protect\ref{Cor: reversible}]
With the function space being a Hilbert space, we are able to introduce the
adjoint $A^{\ast }$ of the operator $A$ with corresponding semigroup $%
E_{t}^{\ast }=e^{A^{\ast }t}$. If $x_{t}$ indeed is reversible in the sense
that $A=A^{\ast }$ then $\sigma \left( A\right) \subseteq (-\infty ,0]$\ and
so (\ref{eq: sigma(A) cond}) is satisfied. (c.f. eq. 5.8 in \cite%
{Scheinkman1995}). Moreover, by the Spectral Mapping Theorem (\cite%
{rudin1973}, Theorem 10.28), the spectrum of the resolvent satisfies%
\begin{equation*}
\sigma \left( R\left( \lambda \right) \right) \backslash \left\{ 0\right\}
=\left( \lambda -\sigma \left( A\right) \right) ^{-1}=\left\{ \frac{1}{%
\lambda -w}:w\in \sigma \left( A\right) \right\}
\end{equation*}%
Since $A$ is self-adjoint so is $R\left( \lambda \right) $ for any $\lambda
\notin \sigma \left( A\right) $. Thus, 
\begin{equation*}
\left\Vert R\left( \lambda \right) \right\Vert _{\mathrm{op}}=\max_{w\in
\sigma \left( R(\lambda \right) )}\left\vert w\right\vert =\max_{w\in \sigma
\left( A\right) }\frac{1}{\left\vert \lambda -w\right\vert }\leq \max_{w\leq
0}\frac{1}{\left\vert \lambda -w\right\vert }=\frac{1}{\left\vert \lambda
\right\vert },
\end{equation*}%
and so (\ref{eq: R(z) cond}) is satisfied.
\end{proof}

\bigskip

\begin{proof}[Proof of Theorem \protect\ref{Th: analytic T 1}]
For any $f\in \mathcal{F}_{0}$,%
\begin{eqnarray*}
\left\Vert A_{D}f\right\Vert _{\mathcal{F}_{0}} &\leq
&\sum_{i=1}^{d}\left\Vert \mu _{i}\right\Vert _{\mathcal{F}}\left\Vert \frac{%
\partial f}{\partial x_{i}}\right\Vert _{\mathcal{F}}+\frac{1}{2}%
\sum_{i,j=1}^{d}\left\Vert \sigma _{ij}^{2}\right\Vert _{\mathcal{F}%
}\left\Vert \frac{\partial ^{2}f}{\partial x_{i}\partial x_{j}}\right\Vert _{%
\mathcal{F}} \\
&\leq &\left( \sum_{i=1}^{d}\left\Vert \mu _{i}\right\Vert _{\mathcal{F}%
_{0}}+\frac{1}{2}\sum_{i,j=1}^{d}\left\Vert \sigma _{ij}^{2}\right\Vert _{%
\mathcal{F}_{0}}\right) \left\Vert f\right\Vert _{\mathcal{F}_{0}} \\
&=&:\bar{A}\left\Vert f\right\Vert _{\mathcal{F}_{0}},
\end{eqnarray*}%
where $\bar{A}<\infty $ under the assumptions of the theorem. Thus, $%
\left\Vert A\right\Vert _{\mathrm{op}}=\sup_{\left\Vert f\right\Vert _{%
\mathcal{F}_{0}}\leq 1}\left\Vert Af\right\Vert _{\mathcal{F}_{0}}<\infty $
and so $A:\mathcal{F}_{0}\mapsto \mathcal{F}_{0}$ is a bounded operator.
This in turn implies that $\sum_{m=0}^{\infty }\frac{t^{m}}{m!}A^{m}f\left(
x\right) $ is a well-defined representation of $w_{t}\left( x\right) $ for
any $f\in \mathcal{F}_{0}$ and so the power series approximation is
consistent. In particular,%
\begin{equation*}
\left\Vert \hat{w}_{t}-w_{t}\right\Vert _{\mathcal{F}_{0}}\leq \frac{t^{M+1}%
}{\left( M+1\right) !}\left\Vert A^{M+1}f\right\Vert _{\mathcal{F}_{0}}\leq 
\frac{\left( t\bar{A}\right) ^{M+1}}{\left( M+1\right) !}\left\Vert
f\right\Vert _{\mathcal{F}_{0}}.
\end{equation*}
\end{proof}

\bigskip

\begin{proof}[Proof of Theorem \protect\ref{Th: analytic bounded}]
The first part follows from Theorem 1.1 in \cite{Escauriaza2017}. For the
second part, First note that $w_{0,t}\left( x\right) =E_{0,t+\tau
_{0}}g\left( x\right) $. Now, by Theorem 1.1 in \cite{Escauriaza2017}, $%
\left\vert \left. \partial _{t}^{m}w_{0,t}\right\vert _{t=0}\right\vert \leq
C\left( \rho \tau _{0}\right) ^{-m}m!$, for all $x\in \mathcal{X}_{0}$, for
some constant $\rho =\rho \left( B,d\right) \in (0,1]$. This in turn implies
that the power series expansion will converge with radius of convergence
bounded by 
\begin{equation*}
T_{0}^{-1}=\lim \sup_{m\rightarrow \infty }\left\{ \frac{1}{m!}\left\Vert
\left. \partial _{t}^{m}w_{0,t}\right\vert _{t=0}\right\Vert _{\mathcal{F}%
,0}\right\} ^{1/m}\leq \frac{1}{\rho \tau _{0}}.
\end{equation*}
\end{proof}

\end{document}